\documentclass{article}

\usepackage{microtype}
\usepackage{graphicx}
\usepackage{subfigure}
\usepackage{booktabs} 

\usepackage{hyperref}

\usepackage[accepted]{icml2024}

\usepackage{amsmath}
\usepackage{amssymb}
\usepackage{mathtools}
\usepackage{amsthm,bm}

\usepackage[capitalize,noabbrev]{cleveref}

\theoremstyle{plain}
\newtheorem{theorem}{Theorem}[section]

\newtheorem{lemma}[theorem]{Lemma}

\theoremstyle{definition}

\theoremstyle{remark}
\newtheorem{remark}[theorem]{Remark}

\newcommand{\av}{{\bf a}}

\newcommand{\xv}{{\bf x}}
\newcommand{\Xv}{{\bf X}}
\newcommand{\uv}{{\bf u}}
\newcommand{\zv}{{\bf z}}

\newcommand{\thetav}{\bm{\theta}}

\newcommand{\Prob}{\mathbb{P}}

\newcommand{\yv}{{\bf y}}
\newcommand{\diag}{{\rm diag}}
\newcommand{\Cc}{\mathcal{C}}
\newcommand{\argmin}{{\rm argmin}}
\newcommand{\E}{\mathbb{E}}
\newcommand{\Ec}{\mathcal{E}}
\newcommand{\Nc}{\mathcal{N}}

\newcommand{\wv}{{\bf w}}

\newcommand{\xvh}{\hat{\xv}}
\newcommand{\xvt}{\tilde{\xv}}
\newcommand{\Xt}{\tilde{X}}

\usepackage{multirow}

\usepackage[textsize=tiny]{todonotes}

\icmltitlerunning{Bagged Deep Image Prior for Recovering Images in the Presence of Speckle Noise}

\begin{document}
\onecolumn

\icmltitle{Bagged Deep Image Prior for Recovering Images in the Presence of Speckle Noise}



\icmlsetsymbol{equal}{*}

\begin{icmlauthorlist}
\icmlauthor{Xi Chen}{yyy}
\icmlauthor{Zhewen Hou}{comp}
\icmlauthor{Christopher A. Metzler}{sch}
\icmlauthor{Arian Maleki}{comp}
\icmlauthor{Shirin Jalali}{yyy}
\end{icmlauthorlist}

\icmlaffiliation{yyy}{Department of Electrical and Computer Engineering, Rutgers University, Piscataway, USA.}
\icmlaffiliation{comp}{Department of Statistics, Columbia University, New York, USA.}
\icmlaffiliation{sch}{Department of Computer Science,
University of Maryland, College Park, USA}

\icmlcorrespondingauthor{Shirin Jalali}{shirin.jalali@rutgers.edu}

\icmlkeywords{Machine Learning}

\vskip 0.3in



\printAffiliationsAndNotice{}  

\begin{abstract}
We investigate both the theoretical and algorithmic aspects of likelihood-based methods for recovering a complex-valued signal from multiple sets of measurements, referred to as looks, affected by speckle (multiplicative) noise. Our theoretical contributions include establishing the first existing theoretical upper bound on the Mean Squared Error (MSE) of the maximum likelihood estimator under the deep image prior hypothesis. Our theoretical results capture the dependence of MSE upon the number of parameters in the deep image prior, the number of looks, the signal dimension, and the number of measurements per look. On the algorithmic side, we introduce the concept of bagged Deep Image Priors (Bagged-DIP) and integrate  them with projected gradient descent. Furthermore, we show how employing Newton-Schulz algorithm for calculating matrix inverses within the iterations of PGD reduces the computational complexity of the algorithm. We will show that this method achieves the state-of-the-art performance. 


\end{abstract}

\section{Introduction}\label{sec:intro}
One of the most fundamental and challenging issues many coherent imaging systems face is the existence of speckle noise. An imaging system with ``fully-developed" speckle noise can be modeled as
\begin{align}\label{eq:firstmodel}
\yv = A X_o \wv + \zv.
\end{align}
Here, $X_o$ denotes a  diagonal matrix with diagonal elements $\xv_o \in \mathbb{C}^n$, representing the complex-valued signal of interest, $\wv \in \mathbb{C}^n$, which is known as speckle noise (or multiplicative noise) includes independent and identically distributed (iid) complex-valued random variables with $\wv_i \sim \mathcal{CN} (0,\sigma_w^2 I_n)$, and finally  $\zv \in \mathbb{C}^m$ denotes the additive noise that is often caused by the sensors and is modeled as iid $\mathcal{CN}(0, \sigma_z^2)$. In this paper, we explore the scenario where $m \leq n$, allowing imaging systems to capture higher resolution images than constrained by the number of sensors.\footnote{Considering $m<n$ for simpler imaging systems (with no speckle noise) has led to the development of the fields of compressed sensing and compressive phase retrieval. }

As is clear from \eqref{eq:firstmodel}, the multiplicative nature of the speckle noise poses a challenge in extracting accurate information from measurements, especially when the measurement matrix $A$ is ill-conditioned. To alleviate this issue, many practical systems employ a technique known as multilook or multishot \cite{argenti2013tutorial,bate2022experimental}. Instead of taking a single measurement of the image, multilook systems capture multiple measurements, hoping for each group of measurements to have independent speckle and additive noises. A model commonly used for multilook systems is:
\[
\yv_l = A X_o \wv_l + \zv_l,
\]
where $l=1,\ldots,L$, with $L$ denoting the number of looks. Moreover,  $\wv_1, \ldots, \wv_L \in \mathbb{C}^n$ and $\zv_1, \ldots, \zv_L \in \mathbb{C}^m$ denote  the  speckle noise vectors and additive noise vectors, respectively. In this model, we have assumed that the measurement kernel
$A$ remains constant across the looks. This assumption holds in multilooking for several imaging systems, such as when sensors' locations change slightly for different looks.

Given the fact that fully-developed noises are complex-valued Gaussian and have uniform phases, the phase of $\xv_o$ cannot be recovered. Hence, the goal of a multilook system is to obtain a precise estimate of $|\xv_o|$, where $|\cdot|$ is the absolute value function applied component-wise to the elements of $\xv_o$, based on the observations $(\yv_1, \ldots, \yv_L)$ and measurement matrix $A$. Given the fact that the phase of $\xv_o$ is not recoverable, in the rest of the paper, we assume that $\xv_o$ is a real-valued signal.

A standard approach for estimating $\xv_o$ is to minimize the negative log-likelihood function:
\begin{equation}\label{eq:formulation1}
\xvh = \arg\min_{\xv \in \Cc} f_L(\xv),
\end{equation}
 where $\Cc$ represents the set encompassing all conceivable images and $f_{L}(\xv)$ is given by:
\begin{align}
    f_{L}(\xv) 
    =&  \log \det(B(\xv)) + \frac{1}{L} \sum_{\ell=1}^L \tilde{\yv}_\ell^T (B(\xv))^{-1} \tilde{\yv}_\ell, \label{eq:ll-SL}
\end{align}
where  
\begin{align*}
    B(\xv) = \begin{bmatrix} \sigma_z^2 I_n + {\sigma_w^2} \Re (A X^2 \bar{A}^T) & -{\sigma_w^2} \Im (A X^2 \bar{A}^T) \\ {\sigma_w^2} \Im (A X^2 \bar{A}^T) & \sigma_z^2 I_n + {\sigma_w^2} \Re (A X^2 \bar{A}^T) \end{bmatrix},
\end{align*}
and
$\tilde{\yv}_\ell^T  = \begin{bmatrix}
        \Re (\yv_\ell^T) & \Im (\yv_\ell^T)
    \end{bmatrix}
$ with $\Re$ and $\Im$ denoting the real and imaginary parts of vectors. Appendix \ref{app:MLE} presents the derivation of this likelihood and its gradient.

It is important to note that the set $\Cc$ in \eqref{eq:formulation1} is not known in practice. Hence, in this paper we work with the following hypothesis that was put forward in \citep{ulyanov2018deep, heckel2018deep}.
\vspace{-.3cm}

\begin{itemize}
\item \textbf{Deep image prior (DIP) hypothesis} ~\citep{ulyanov2018deep, heckel2018deep}: Natural images can be embedded within the range of neural networks that have substantially fewer parameters than the total number of pixels, and use iid noises as inputs. 
\end{itemize}
\vspace{-.3cm}

Inspired by this hypothesis, we consider $\Cc$ in \eqref{eq:formulation1} as the range of a deep image prior. More specifically, we assume that for every $\xv \in \Cc$ there exists $\thetav \in \mathbb{R}^k$ such that $ \xv \approx g_{\thetav} (\uv)$, where $\uv$ is generated iid from distribution $\mathcal{N}(0,1)$, and $\thetav \in \mathbb{R}^k$ denotes the parameters of the neural network. There are two main challenges that we address: 
\vspace{-.3cm}

\begin{itemize}
\item Theoretical challenge: Assuming we can solve the optimization problem \eqref{eq:formulation1} under the DIP hypothesis, the question arises: can we theoretically assess the quality of the reconstruction? Additionally, what is the behavior of the reconstruction error in relation to $k$, number of parameters of the neural network, $m, n,$ and $ L$? Specifically, in the scenario where the scene is static, and we have the ability to acquire as many looks as necessary, what level of accuracy can be expected in the final estimate?
\vspace{-.2cm}

\item Practical challenge: Given the challenging nature of the likelihood and the DIP hypothesis, what is a good way to solving \eqref{eq:formulation1} under the DIP hypothesis?
\end{itemize}
\vspace{-.3cm}

Here is a summary of our contributions:

On the theoretical side, we will establish the first theoretical result regarding the performance of multilook coherent imaging systems. Our theoretical results unveil intriguing characteristics of such imaging systems. For instance, we will show that when $L$ is large, these systems offer very accurate reconstructions when $m^2 = \omega (k \log n)$. We will clarify some of our novel technical contributions for obtaining such a theoretical result in Section \ref{sec:main:theory}. 

On the applied side, we start with vanilla projected gradient descent (PGD) \cite{lawson1995solving}, but it faces two challenges diminishing its effectiveness on this problem:

\begin{enumerate}

\item[] Challenge 1: As will be described in Section \ref{sec:PGDchallenges}, in the PGD, the signal to be projected on the range of $g_{\thetav}(\uv)$ is burried in ``noise''. Hence, deep image priors with large number of parameters will overfit to the noise and will not let the PGD algorithm obtain a better estimate \cite{heckel2018deep,heckel2019denoising}. On the other hand, the low accuracy of simpler DIPs become a bottleneck as the algorithm progresses through iterations, impacting the overall performance of PGD. To alleviate this issue, we propose \textbf{Bagged-DIP}. This is a simple idea with roots in classical literature of ensemble methods \cite{breiman1996bagging}. Bagged-DIP idea enables us to use complex DIPs at every iteration and yet obtain accurate results. 

\item[] Challenge 2: As will be clarified in Section \ref{ssec:pgd:dip:intro}, PGD requires the inversion of large matrices at every iteration, which is a computationally challenging problem. We alleviate this issue by using the Newton-Schulz algorithm \cite{schulz1933iterative}, and empirically demonstrating that \textbf{only one} step of this algorithm is sufficient for the PGD algorithm. This significantly reduces the computational complexity of each iteration of PGD. 

\end{enumerate}

\section{Related work}

Eliminating speckle noise has been extensively explored in the literature~\cite{lim1980techniques,gagnon1997speckle,tounsi2019speckle}. Current technology relies on gathering enough measurements to ensure the invertibility of matrix $A$ and subsequently inverting $A$ to represent the measurements in the following form: $\yv_\ell = X\wv_\ell+ \zv_\ell$.
However, as matrix $A$ deviates from the identity, the elements of the vector $\zv$ become dependent. In practice, these dependencies are often overlooked, simplifying the likelihood. This simplification allows researchers to leverage various denoising methods, spanning from classical low-pass filtering to application of convolutional neural networks \cite{tian2020attention} and transformers \cite{fan2022sunet}. A series of papers have considered the impact of the measurement kernel in the algorithms. By using single-shot digital holography, the authors in~\cite{pellizzari2017phase, pellizzari2018optically} develop heuristic method to obtain maximum a posteriori estimate of the real-valued speckle-free object reflectance. They later extend this method to handle multi-shot measurements and incorporate more accurate image priors~\cite{pellizzari2020coherent,pellizzari2022solving, bate2022experimental}. While these methods can work with non-identity $A$'s, they still require $A$ to be well-conditioned.

Our paper is different from the existing literature, mainly because we study scenarios where the matrix $A$ is under-sampled ($m<n$). In a few recent papers, researchers have explored similar problems~\cite{zhou2022compressed,chen2023multilook}. The paper \cite{chen2023multilook} aligns closely in scope and approach with our work. The authors addressed a similar problem, albeit assuming real-valued measurements and noises, and advocated for the use of DIP-based PGD. Addressing the concerns highlighted in the last section (further elucidated in Section \ref{sec:simulation_baggedDIP}), our Bagged-DIP-based PGD employing the Newton-Schulz algorithm significantly outperforms \cite{chen2023multilook} in both reconstruction quality and computational complexity. We will provide more information in our simulation studies. 
Furthermore, we should emphasize that \cite{chen2023multilook} did not offer any theoretical results regarding the performance of DIP-based MLE.  

 The authors in in \cite{zhou2022compressed} theoretically demonstrated the feasibility of accurate estimates for $\mathbf{x}_o$ even with $m<n$ measurements. While our theoretical results build upon the contributions of \cite{zhou2022compressed}, our paper extends significantly in two key aspects: (1) We address the multilook problem and investigate the influence of the number of looks on our bounds. To ensure sharp bounds, especially when $L$ is large, we derive sharper bounds than those presented in \cite{zhou2022compressed}. These require novel technical contributions (such as using decoupling method)  as detailed in our proof. (2) In contrast to the use of compression codes' codewords for the set $\Cc$ in \cite{zhou2022compressed}, we leverage the range of a deep image prior, inspired by recent advances in machine learning. Despite presenting new challenges in proving our results, this approach enables us to simplify and establish the relationship between Mean Squared Error (MSE) and problem specification parameters such as $n, m, k, L$.

Given DIP's flexibility, it has been employed for various imaging and (blind) inverse problems, e.g., compressed sensing, phase retrieval etc.~\cite{jagatap2019algorithmic, ongie2020deep,darestani2021accelerated, ravula2022one,zhuang2022practical,zhuang2023blind}. 
To boost the performance of DIP in these applications, researchers have explored several ideas, including, introducing explicit regularization ~\citep{mataev2019deepred}, incorporating prior on network
weights by introducing a learned regularization method into the DIP structure \cite{van2018compressed}, and exploring the effect of changing DIP structures and input noise settings to speed up DIP training \cite{li2023deep}. 


Lastly, it's important to note our work can be situated within the realm of compressed sensing (CS) \cite{donoho2006compressed, candes2008introduction, davenport2012introduction, bora2017compressed, peng2020solving, joshi2021plugin, nguyen2022provable}, where the objective is to derive high-resolution images from lower-resolution measurements. However, notably, the specific challenge of recovery in the presence of speckle noise has not been explored in the literatures before, except in \cite{zhou2022compressed} that we discussed before.


\section{Main theoretical result}\label{sec:main:theory}

As we described in the last section, in our theoretical work, we consider the cases in which $m <n$. $m$ can  even be much smaller than $n$. Furthermore, for notational simplicity, in our theoretical work only, we assume that the measurements and noises are real-valued.\footnote{For the complex-valued problem, since the phases of the elements of $\xv_o$ are not recoverable, we can assume that $\xv_o$ is real-valued. Even though in this case, the problem is similar to the problem we study in this paper, given that we have to deal with real and imaginary parts of the measurement matrices and noises, they are notationally more involved.} Hence, we work with the following likelihood function:
\begin{equation}\label{eq:formulation2}
\xvh = \arg\min_{\xv \in \Cc} f(\xv),
\end{equation}
where 
\begin{align}\label{eq:updated_lik}
 f(\xv) = \log\det \left(\sigma_z^2 I_m + \sigma_w^2 AX^2A^T\right) + \frac{1}{L} \sum_{\ell} \yv_{\ell}^T \left( \sigma_z^2I_m+\sigma_w^2 AX^2A^T\right)^{-1}  \yv_{\ell}. 
\end{align}
Note that we omit subscript $L$ from the likelihood as a way to distinguish between the negative loglikelihood of real-valued measurements from the complex-valued ones. The following theorem is the main theoretical result of the paper. Consider the case of no additve noise, i.e. $\sigma_z =0$, and that for all $i$, we have $0 < x_{\min} \leq x_{o,i} \leq x_{\max}$.

\begin{theorem}\label{thm:maintheorem}
    Let the elements of the measurement matrix $A_{ij}$ be iid $\Nc(0,1)$. Suppose that $m<n$ and that the function $g_{\thetav} (\uv)$, as a function of $\thetav \in [-1,1]^{k}$, is Lipschitz with Lipschitz constant $1$. We have
    \begin{align}\label{eq:thm:MSEbound}
    \frac{1}{n} \|\xvh-\xv_o\|_2^2 =  O\left( \frac{ \sqrt{k \log n }}{m} +  \frac{n \sqrt{ k \log n}}{  m \sqrt{L m}}  \right),
    \end{align}
    with probability $1- O(e^{-\frac{m}{2}} + e^{-\frac{Ln}{8}} + e^{-k \log n}+   {\rm e}^{k \log n-\frac{n}{2}})$.
\end{theorem}

Before we discuss the proof sketch and the technical novelties of our proof strategy, let us explain some of the conclusions that can be drawn from this theorem and provide some intuition. As is clear in \eqref{eq:thm:MSEbound}, there are two terms in the mean square error. One that does not change with $L$ and the other term that decreases with $L$. To understand these two terms, we provide further explanation in the following remarks.

\begin{remark}
As the number of parameters of DIP, $k$, increases (while keeping $m, n$, and $L$ fixed), both error terms in the upper bound of MSE grow. This aligns with intuition, as increasing the number of parameters in $g_{\thetav}(\uv)$ allows the DIP model to generate more intricate images. Consequently, distinguishing between these diverse alternatives  based on the measurements becomes more challenging.
\end{remark}

\begin{remark}
The main interesting feature of the second term in the mean square error, i.e. $\frac{n \sqrt{ k \log(n)}}{  m \sqrt{L m}}$  is the fact that it grows rapidly as a function of $n$. In imaging systems with additive noise only, the growth is often logarithmic in $n$ \cite{bickel2009simultaneous}, contrasting with polynomial growth observed in here. This can be associated with the fact that as we increase $n$ we increase the number of speckle noise elements present in our measurements. Hence, it is reasonable to expect the error term to grow faster in $n$ compared with additive noise models. But the exact rate at which the error terms increases is yet unclear. As will be clarified in the proof, most of the upper bounds we derive are expected to be sharp (modulo one step in which we use a union bound, and we do not expect that union bound to loosen our upper bounds). Hence, we believe $\frac{n \sqrt{ k \log(n)}}{  m \sqrt{L m}}$ is sharp too.   
\end{remark}

\begin{remark}
As $L \rightarrow \infty$, the second term in the upper bound of MSE converges to zero, and the dominant term becomes $\sqrt{k \log n}/m$. Note that since we are considering a fixed matrix $A$, accross the looks, even when $L$ goes to infinity, we should not expect to be able to recover $\xv_o$ indepedent of the value of $m$. One heuristic way to see this is to calculate 
\begin{align}
\frac{1}{L} \sum_{\ell=1}^L \yv_{\ell} \yv_{\ell}^T =  AX_o \frac{1}{L}\sum_{\ell=1}^L \wv_{\ell} \wv_{\ell}^{T} X_oA^T.
\end{align} 
If we heuristically apply the weak law of large numbers and use the approximation $\frac{1}{L}\sum_{\ell} \wv_{\ell} \wv_{\ell}^{T} \approx I$, we can see that 
\[
\frac{1}{L} \sum_{\ell} \yv_{\ell} \yv_{\ell}^T \approx AX_o^2 A^T.
\]
Under these approximations, the matirx $\frac{1}{L} \sum_{\ell} \yv_{\ell} \yv_{\ell}^T$ provides $m(m+1)/2$ (due to symmetry) linear measurements of $X_o^2$. Hence, inspired by the field of compressed sensing, intuitively speaking, we expect the accurate recovery of $\xv_o^2$ to be possible when $m^2 \gg k \log n$ \cite{jalali2016compression}. The first error term of MSE is negligible when $m^2 \gg k \log n$, which is consistent with our conclusion based on the limit of $\frac{1}{L} \sum_{\ell} \yv_{\ell} \yv_{\ell}^T$. 

\end{remark}

We next provide a brief sketch of the proof to highlight the technical novelties of our proof and also to enable the readers to navigate through the detailed proof  more easily.

\textit{Proof sketch of Theorem \ref{thm:maintheorem}}. 
Let $\hat{\xv}_o$ denote the minimizer of function $f$ defined as 
\begin{align}\label{eq:likelihood}
f(\xv)=f(\Sigma(\xv))=-\log\det \Sigma +{1\over L \sigma_w^2}\sum_{\ell =1}^L{\rm Tr}(\Sigma\yv_{\ell}\yv_{\ell}^T),
\end{align} 
with $ \Sigma=\Sigma(\xv)=(AX^2A^T)^{-1},$ where $X=\diag(\xv)$.  For the reasons that will become clear later, we consider a $\delta_n$-net\footnote{The subscript $n$ of $\delta_n$ emphasizes that $\delta_n$ depends on $n$ and is very close to zero when $n$ is large. } of the set $[-1,1]^k$ and we call the mapping of the $\delta_n$-net under $g$, $\Cc_n$. The choice of $\delta_n$ will be discussed later. Define $\xvt_o$ as the closest vector in $\Cc_n$ to $\xvh_o$, i.e.,
\[
\xvt_o=\argmin_{\xv\in\Cc_n}\|\xvh_o-\xv\|.
\]
Let  $\Xt_o={\rm diag}({\xvt_o})$. 
Define $\Sigma_o$, $\hat{\Sigma}_o$ and $ \tilde{\Sigma}_o$, as $\Sigma_o=\Sigma(\xv_o)$, $\hat{\Sigma}_o=\Sigma(\hat{\xv}_o)$, $ \tilde{\Sigma}_o=\Sigma({\xvt}_o)$, respectively. Since $\xvh_o$ is the minimizer of \eqref{eq:likelihood}, we have
\begin{align}
f(\hat{\Sigma}_o)\leq f({\Sigma}_o). \label{eq:thm1-step1}
\end{align}

On the other hand, $f$ can be written as $f(\Sigma)=-\log\det \Sigma +{1\over L\sigma_w^2} \sum_{\ell=1}^L{\rm Tr}(\Sigma AX_o\wv_\ell \wv^T_\ell X_oA^T)$. Let $\bar{f}(\Sigma)$ denote the expected value of  ${f}(\Sigma)$ with respect to $\wv_1, \ldots, \wv_{\ell}$. It is straightforward to show
\begin{align}
\bar{f}(\Sigma)=-\log\det \Sigma +{\rm Tr}(\Sigma AX_o^2A^T).\label{eq:min-f-bar}
\end{align}
As a function of $\Sigma$, $\bar{f}$ achieves its minimum at $\Sigma_o^{-1}= AX_o^2A^T$. We have
\begin{align}\label{eq:first_ineq}
\bar{f}(\tilde{\Sigma}_o)-\bar{f}(\Sigma_o)&=\bar{f}(\tilde{\Sigma}_o)- f(\tilde{\Sigma}_o)+f(\tilde{\Sigma}_o)-{f}(\hat{\Sigma}_o)+{f}(\hat{\Sigma}_o)-{f}({\Sigma}_o)+f({\Sigma}_o)-\bar{f}({\Sigma}_o) \nonumber\\
&\leq \bar{f}(\tilde{\Sigma}_o)- f(\tilde{\Sigma}_o)+f(\tilde{\Sigma}_o)-{f}(\hat{\Sigma}_o)+f({\Sigma}_o)-\bar{f}({\Sigma}_o),
\end{align}
where to obtain the last inequality we have used \eqref{eq:thm1-step1}. The roadmap of the rest of the proof is the following:
\begin{enumerate}
\item Obtaining a lower bound for $\bar{f}(\tilde{\Sigma}_o)-\bar{f}(\Sigma_o)$ in terms of $\|\xvt_o- \xv_o\|_2^2$. Note that since $\bar{f} (\Sigma)$ is a convex function of $\Sigma$ and is minimized at $\Sigma_o$ we expect to be able to obtain such bounds. Nevertheless this is the most challenging part of the proof, because of the relatively complicated dependence of $\xv$ and $\Sigma$, and the dependence of $\Sigma$ on $A$ in addition to $\xv$. Using sharp linear-algebraic bounds combined with the decoupling ideas \cite{de2012decoupling} enabled us to obtain a sharp lower bound for this quantity.  

\item Finding upper bounds for $\bar{f}(\tilde{\Sigma}_o)- f(\tilde{\Sigma}_o)$ and $f (\Sigma_o) - \bar{f}({\Sigma}_o)$. Such bounds can be obtained using standard concentration of measure results, such as Hanson-Wright inequality, and the concentration of singular values of iid Gaussian random matrices. 

\item Finding an upper bound for $f(\tilde{\Sigma}_o)-{f}(\hat{\Sigma}_o)$: note that intuitively, we expect $\|\tilde{\Sigma}_o- \hat{\Sigma}_o\|$ to be small as well. Assuming that function $f$ is a nice function, in the sense that it maps nearby points to nearby points in its range, we expect $f(\tilde{\Sigma}_o)-{f}(\hat{\Sigma}_o)$ to be small too. However, note that the function $f$ has the randomness of $\wv_1, \ldots, \wv_L$ and $A$. Hence, to make our heuristic argument work, we have to first prove that with high probability $f$ is a nice function.  
\end{enumerate}
Details of the three steps is presented in Section \ref{ssec:mainstepsappendix}.

\section{Main algorithmic  contributions}\label{ssec:mainempirical}
\subsection{Summary of projected gradient descent and DIP}\label{ssec:pgd:dip:intro}
As discussed in Section \ref{sec:intro}, we aim to solve the optimization problem \eqref{eq:updated_lik} under the DIP hypothesis. A popular heuristic for achieving this is using projected gradient descent (PGD). At each iteration $t$, the estimate $\mathbf{x}^t$ is updated as follows:
\begin{equation}
    {\xv}^{t+1} = \text{Proj} (\xv^t - \mu_t \nabla \ell({\xv}^t)),
\end{equation}
where $\text{Proj} (\cdot)$ projects its input onto the range of the function $g_{\thetav} (\uv)$, and $\mu_t$ denotes the learning rate. The details of the calculation of $\nabla \ell(\mathbf{x}^t)$ are outlined in Appendix \ref{app:MLE}.

An outstanding question in the implementation pertains to the nature of the projection operation $\text{Proj} (\cdot)$. If $g_{\theta}(\uv)$, in which $\theta$ denotes the parameters of the neural network and $\uv$ denotes the input Gaussian noise, represents the reconstruction of the DIP, during training, DIP learns to reconstruct images by performing the following two steps:
\begin{align}
&\hat{\theta}^t = \operatorname*{argmin}_{\theta} \| g_{\theta}(\uv) - (\xv^t - \mu_t \nabla \ell({\xv}^t)) \|, \nonumber \\
&{\xv}^{t+1} = g_{\hat{\theta}^t}(\uv), 
\end{align}
where to obtain a local minima in the first optimization problem, we use Adam \cite{kingma2014adam} .  
One of the main challenges in using DIPs in PGD is that the performance of DIP $g_{\thetav}(\uv)$ is affected by the structure choices, training iterations as well as the statistical properties of $(\xv^t - \mu_t \nabla \ell({\xv}^t))$~\cite{heckel2019denoising}. We will discuss this issue in the next section. 

\begin{figure}[t]
\centering
\includegraphics[width=0.4\textwidth]{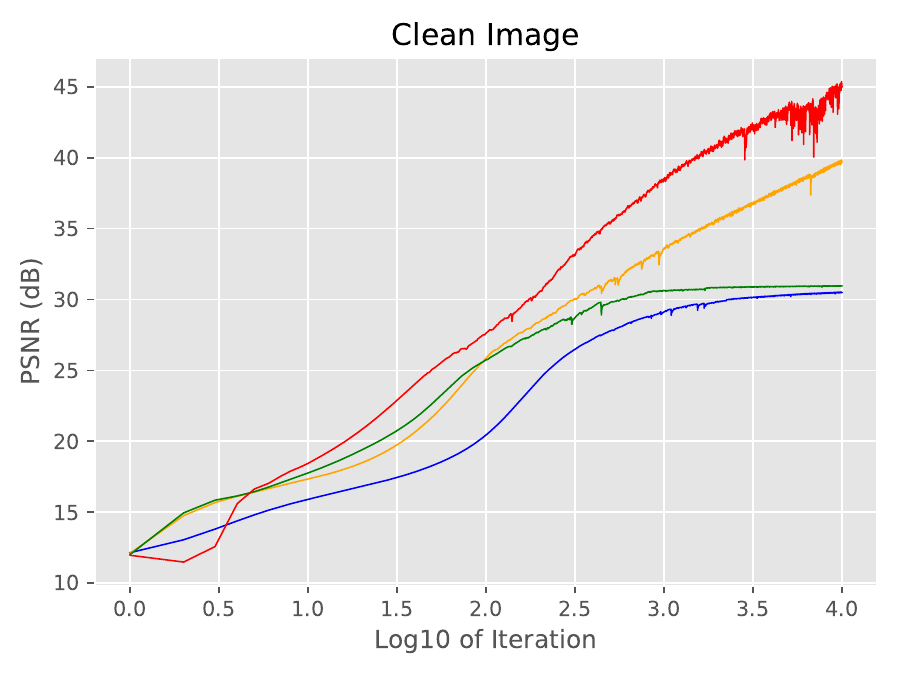}
\includegraphics[width=0.4\textwidth]{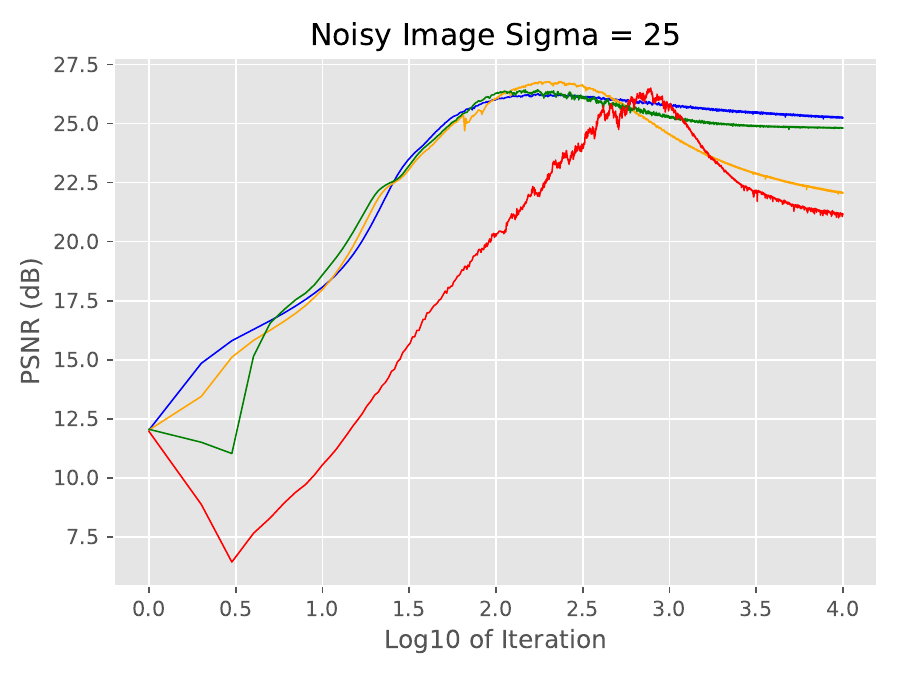}
\vspace{-.3cm}
\caption{PSNR (averaged over 8 images) versus iteration count is depicted for four DIP models fitted to both clean (left panel) and noisy images with noise level $\sigma=25$ (right panel). The 4-layer networks are specified as follows: Blue - kernel size=1, channels [100, 50, 25, 10]; Orange - kernel size=3, channels same as Blue; Green - kernel size=1, channels [128, 128, 128, 128]; Red - kernel size=3, channels same as Green.}
\label{fig:decoder_compare_clean}
\end{figure}
\subsection{Challenges of DIP-based-PGD} \label{sec:PGDchallenges}

In this section, we examine two primary challenges encountered by DIP-based PGD and present novel perspectives for addressing them.

\subsubsection{Challege 1: Right choice of DIP}

Designing PGD, as described in Section \ref{ssec:pgd:dip:intro}, is particularly challenging when it comes to selecting the appropriate network structure for DIP. Figure \ref{fig:decoder_compare_clean} clarifies the main reason. In this figure, four DIP networks are used for fitting to the clean image (left panel) and an image corrupted by the Gaussian noise (right panel). As is clear, the sophisticated networks fit the clean image very well. However, they are more susceptible to overfitting when the image is corrupted with noise. On the other hand, the networks with simpler structure do not fit to the clean image well, but are less susceptible to the noise than the sophisticated-DIPs. This issue has been observed in previous work \cite{heckel2018deep, heckel2019denoising}. 

The problem outlined above poses a challenge for the DIP-based PGD. Note that if $\mathbf{x}^t - \mu_t \nabla f(\mathbf{x}^t)$ closely approximates $\mathbf{x}_o$, fitting a highly intricate deep decoder to $\mathbf{x}^t - \mu_t \nabla f(\mathbf{x}^t)$ will yield an estimate that remains close to $\mathbf{x}_o$. Conversely, if overly simplistic networks are employed in this scenario, their final estimate may fail to closely approach $\mathbf{x}^t - \mu_t \nabla f(\mathbf{x}^t)$, resulting in a low-quality estimate. In the converse scenario, where $\mathbf{x}^t - \mu_t \nabla f(\mathbf{x}^t)$ is significantly distant from $\mathbf{x}_o$, a complex network may overfit to the noise. On the contrary, a simpler network, capable of learning only fundamental features of the image,  may generate an estimate that incorporates essential image features, bringing it closer to the true image.

The above argument suggests the following approach: initiate DIP-PGD with simpler networks and progressively shift towards more complex structures as the estimate quality improves\footnote{A somewhat weaker  approach would be to use intricate networks at every iteration, but then use some regularization approach such as early stopping to control the complexity of the estimates. }. However, finding the right complexity level of the deep image prior for each iteration of PGD, in which the statistics of the error in the estimate $\mathbf{x}^t - \mu_t \nabla f(\mathbf{x}^t)$ is not known and may be image dependent, is a challenging problem. In the next section, we propose a new approach for aliviating this problem. 

\subsubsection{Solution to Challege 1: Bagged-DIP}\label{ssec:baggedDIP}

Our new approach is based on a classical idea in statistics and machine learning: Bagging. 
Rather than finding the right complexity level for the DIP at each iteration, which is a computationally demanding and statistically challenging problem, we use  bagging. The idea of bagging is that in the case of challenging estimation problems, we create several low-bias and hopefully weakly dependent estimates  (we are overloading the phrase weakly-dependent to refer to situations in which the cross-correlations of different estimates are not close to $1$ or $-1$)  of a quantity and then calculate the average of those quantities to obtain a lower-variance estimate. In order to obtain weakly dependent estimates, a common practice in the literatire is to apply the same learning scheme to multiple datasets, each of which is a random purterbation of the original training set, see e.g. the construction of random forrests. 

While there are many ways to create Bagged-DIP estimates, in this paper, we explore a few very simple estimates, leaving other choices for future research. First we select a network that is sophisticated enough to fit well to real-world images. The details of the newtwork we use for this paper can be found in Appendix~\ref{app:add_experiment}. Using the neural network provides our initial estimate of the image from the noisy observation.
%
%
To generate a new estimate, we begin by selecting an integer number $k$, partitioning an image of size $(H \times W)$ into non-overlapping patches of sizes $(h_k \times w_k)$. Independent Deep Image Priors (DIPs), with the same structure as the main one, are then employed to reconstruct each of these $(h_k \times w_k)$ patches. Essentially, the estimation of the entire image involves learning $\frac{HW}{h_k w_k}$ DIP models. By placing these $\frac{HW}{h_k w_k}$ patches back into their original positions, we obtain the estimate of the entire image, denoted as $\check{\xv}_k$. A crucial aspect of this estimate is that the estimation of a pixel relies solely on the $h_k \times w_k$ patch to which the pixel belongs and no other pixel. By iterating this process for $K$ different values of $h_k \times w_k$, we derive $K$ estimates denoted as $\check{\xv}_1, \ldots, \check{\xv}_K$. The final sought-after estimate is obtained by averaging the individual estimates. 

The estimation of a pixel in $\check{\xv}_k$ is only dependent on the $h_k \times w_k$ patch to which the pixel belongs. As our estimates for different values of $k$ utilize distinct regions of the image to derive their pixel estimates, we anticipate these estimates to be weakly-dependent (again in the sense that the cross-correlations are not close to $1$ or $-1$).

\subsubsection{Challenge 2: Matrix inversion}

The gradient of $f_L(\xv)$ defined in \eqref{eq:formulation1} is (see Appendix \ref{app:MLE}):
\begin{align} \label{eq:GD-ML}
    \frac{\partial f_L}{\partial x_j}=  2 \xv_j \sigma_w^2 \left( \tilde{\av}_{\cdot, j}^{+T} B^{-1} \tilde{\av}_{\cdot, j}^{+} + \tilde{\av}_{\cdot, j}^{-T} B^{-1} \tilde{\av}_{\cdot, j}^{-} \right) - \frac{2 \xv_j \sigma_w^2}{L } \sum_{l=1}^L \left[ \left( \tilde{\av}_{\cdot, j}^{+T} B^{-1} \tilde{\yv}_l \right)^2 + \left( \tilde{\av}_{\cdot, j}^{-T} B^{-1} \tilde{\yv}_l \right)^2 \right],
\end{align}
where
$\Tilde{\mathbf{a}}^+_{\cdot, j} = \begin{bmatrix} \Re (\mathbf{a}_{\cdot, j}) \\ \Im (\mathbf{a}_{\cdot, j}) \end{bmatrix}$, $\Tilde{\mathbf{a}}^-_{\cdot, j} = \begin{bmatrix} -\Im (\mathbf{a}_{\cdot, j}) \\ \Re (\mathbf{a}_{\cdot, j}) \end{bmatrix}$, $\tilde{\yv}_l = \begin{bmatrix} \Re (\yv_l) \\ \Im (\yv_l) \end{bmatrix}$, $\mathbf{a}_{\cdot, j}$ denotes the $j$-th column of matrix $A$. It's important to highlight that in each iteration of the PGD, the matrix $B$ changes because it depends on the current estimate $\xv^t$. 
This leads to the computation of the inverse of a large matrix $B \in \mathbb{R}^{2m \times 2m}$ at each iteration, posing a considerable computational challenge and a major obstacle in applying DIP-based PGD for this problem.
In the next section, we present a solution to address this issue.

\subsubsection{Solution to Challenge 2}

To address the challenge mentioned in the last section, we propose to use Newton-Schulz algorithm. Newton-Schulz, is an iterative algorithm for obtaining a matrix inverse. The iterations of Newton-Schulz for finding $(B_t)^{-1}$ is given by
\begin{equation}
    M^{k} = M^{k-1} + M^{k-1} (I- B_t M^{k-1}),
\end{equation}
where $M^{k}$ is the approximation of $(B_t)^{-1}$ at iteration $k$. It is shown that if $\sigma_{\max}(I - M^0 B_{t}) < 1$, the Newton-Schulz converges to $B_{t}^{-1}$ quadratically fast \cite{gower2017randomized, stotsky2020convergence}. 


An observation to alleviate the mentioned issue in the previous section is that, given the nature of the gradient descent, we don't anticipate significant changes in the matrix $X_t^2$ from one iteration to the next. Consequently, we expect $B_t$ and $B_{t-1}$, as well as their inverses, to be close to each other.


Hence, instead of calculating the full inverse at iteration $t+1$, we can employ the Newton-Schulz algorithm with $M^0$ set to $(B_t)^{-1}$ from the previous iteration.  
Our simulations will show that \textbf{one} step of the Newton-Schulz algorithm suffices. 

\begin{table*}[t]
    \centering
    \scriptsize
\begin{tabular}{lccccccccccc}\hline
    \textbf{m/n} & \textbf{$\#$looks} & \textbf{Barbara} & \textbf{Peppers} & \textbf{House} & \textbf{Foreman} & \textbf{Boats} & \textbf{Parrots} & \textbf{Cameraman} & \textbf{Monarch} & \textbf{Average}\\\hline
\multirow{3}{*}{12.5\%}  
    & 25 & 19.91/0.443 & 19.70/0.385 & 20.15/0.377 & 19.10/0.355 & 20.20/0.368 & 17.61/0.372 & 18.19/0.426 & 19.06/0.524 & 19.24/0.406 \\
    & 50 & 20.90/0.567 & 21.69/0.535 & 22.27/0.531 & 20.51/0.577 & 21.41/0.470 & 19.23/0.486 & 19.31/0.492 & 21.33/0.642 & 20.83/0.538\\
    & 100 & 21.84/0.633 & 22.41/0.657 & 23.96/0.624 & 20.63/0.638 & 22.52/0.536 & 19.68/0.574 & 20.66/0.512 & 22.56/0.720 & 21.78/0.612\\\hline\hline
\multirow{3}{*}{25\%}  
    & 25 & 23.57/0.586 & 23.17/0.547 & 24.25/0.520 & 23.30/0.526 & 22.77/0.487 & 21.23/0.522 & 21.50/0.496 & 23.13/0.707 & 22.86/0.549\\
    & 50 & 25.38/0.689 & 25.12/0.691 & 26.84/0.652 & 25.12/0.681 & 24.50/0.601 & 23.37/0.636 & 24.30/0.642 & 24.93/0.785 & 24.95/0.672\\
    & 100 & 26.26/0.748 & 26.14/0.759 & 28.33/0.717 & 26.41/0.772 & 25.72/0.682 & 24.55/0.720 & 26.22/0.719 & 26.28/0.845 & 26.24/0.745\\\hline\hline
\multirow{3}{*}{50\%}  
    & 25 & 27.30/0.759 & 27.02/0.724 & 28.56/0.697 & 27.56/0.735 & 26.21/0.669 & 25.94/0.728 & 27.95/0.762 & 27.17/0.845 & 27.21/0.740\\
    & 50 & 28.67/0.816 & 28.52/0.804 & 30.30/0.762 & 28.88/0.827 & 27.58/0.739 & 27.23/0.799 & 30.21/0.843 & 28.86/0.898 & 28.78/0.818\\
    & 100 & 29.40/0.843 & 29.21/0.849 & 31.61/0.815 & 29.74/0.871 & 28.45/0.785 & 28.20/0.848 & 31.58/0.902 & 30.05/0.932 & 29.78/0.856\\\hline
\end{tabular}
    \vspace{-0.1cm}
    \caption{PSNR(dB)/SSIM $\uparrow$ of 8 test images with $m/n=12.5\%/25\%/50\%$, $L=25/50/100$.}
    \label{tab:main-CS-PSNR-SSIM}
\vspace{-0.3cm}
\end{table*}

\section{Simulation results}

\subsection{Study of the impacts of different modules}

\subsubsection{Newton-Schulz iterations}

In this section, we aim to answer the following questions: (1) Is the Newton-Schulz algorithm effective in our Bagged-DIP-based PGD? (2) What is the minimum number of iterations for the Newton-Schulz algorithm to have good performance in Bagged-DIP-based PGD? (3) How does the computation time differ when using the Newton-Schulz algorithm compared to exact inverse computation?

\begin{figure}[ht]
\centering
\includegraphics[width=0.4\textwidth]{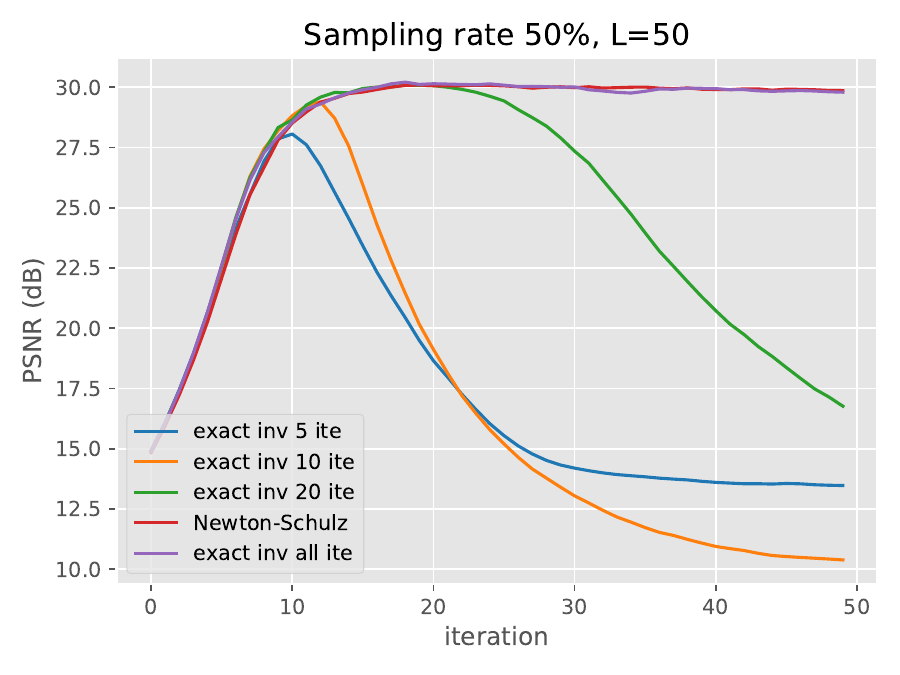}
\vspace{-.4cm}
\caption{
Newton-Schulz approximation (red) compared with computing exact inverse (purple). Blue, orange and green curves correspond to stopping the update of the inverse after the first $5$, $10$, and $20$ iterations respectively.}
\label{fig:iterative_alg}
\end{figure}

Figure \ref{fig:iterative_alg} shows one of the simulations we ran to address the first two questions. In this figure, we have chosen $L=50$ and $m/n =0.5$, and the learning rate of PGD is $0.01$. The result of Bagged-DIP-based PGD with a \textbf{single step} of Newton-Schulz is virtually identical to PGD with the exact inverse. To investigate the impact of the Newton-Schulz algorithm further, we next checked if applying even one step of Newton-Schulz is necessary. Hence, in three different simulations we stopped the matrix inverse update at iterations $5$ (blue), $10$ (orange), and $20$ (green). As is clear from Figure \ref{fig:iterative_alg}, a few iterations after stopping the update, PGD starts diverging. Hence, we conclude that a single step of the Newton-Schulz is necessary and sufficient for PGD. 

To address the last question raised above, we evaluated how much time the calculation of the gradient takes if we use one step of the Newton-Schulz compared to the full matrix inversion. Our results are reported in Table \ref{tab:Newton_Schulz_time}. Our simulations are for sampling rate $50\%,$ and number of looks $ L=50$ and three different images sizes.\footnote{Our algorithm still faces memory limitations on a single GPU when processing $256 \times 256$ images. Addressing this issue through approaches like parallelization remains subject for future research.} As is clear the Newton-Schulz is much faster.  

\vspace{-.3cm}
\begin{table}[ht]
    \centering
        \caption{Time (in seconds) required for exact matrix inversion and its Newton-Schulz approximation in PGD step.}
    \label{tab:Newton_Schulz_time}
    \scriptsize
\begin{tabular}{ccccc}\hline
    & \textbf{Image size} & \textbf{32 $\times$ 32} & \textbf{64 $\times$ 64} & \textbf{128 $\times$ 128} \\\hline
    & \textbf{GD w/ Newton-Schulz} & $\sim$ 7e-5 & $\sim$ 8e-5 & $\sim$ 1e-4\\
    & \textbf{GD w/o Newton-Schulz}  & $\sim$ 0.3 & $\sim$ 1.2 & $\sim$ 52.8\\\hline
\end{tabular}
\end{table}
\vspace{-0.1cm}

In our final algorithm, if the difference between $\|\xv^{t}-\xv^{t-1}\|_{\infty} > \delta_{\xv}$, then we use the exact inverse update. $\delta_{\xv}$ is set to 0.12 (please refer to Appendix \ref{app:add_experiment} for details) in all our simulations. Based on this updating criterion, we observe that the exact matrix inverse is only required for the first 2-3 iterations, and it is adaptive enough to guarantee the convergence of PGD.

\subsubsection{Bagged-DIP}
Intuitively speaking, the more weakly dependent estimates one generate the better the average estimate will be. In the context of DIPs, there appear to be many different ways to create weakly dependent samples. The goal of this section is not to explore the full-potential of Bagged-DIPs. Instead, we aim to demonstrate that even a few weakly dependent samples can offer noticeable improvements. Hence, unlike the classical applications of bagging in which thousands of bagged samples are generated, to keep the computations managable, we have only considered three bagged estimates. Figure \ref{fig:simple_bagged_L} shows one of our simulations. More simulations are in Appendix \ref{app:add_exp_bagging}. In this simulation we have chosen $K=3$, i.e. we have only three weakly-dependent estimates. These estimates are constructed according to the recipe presented in Section \ref{ssec:baggedDIP} with the following patch sizes: 
$h_1 = w_1 = 32$, $h_2 = w_2 = 64$, and $h_3 = w_3 = 128$.
As is clear from the left panel of Figure \ref{fig:simple_bagged_L}, even with these very few samples, Bagged-DIPs has offered between $0.5$dB and $1$dB  over the three estimates it has combined.  

\vspace{-.1cm}
\begin{figure}[ht]
\centering
\includegraphics[width=0.4\textwidth]{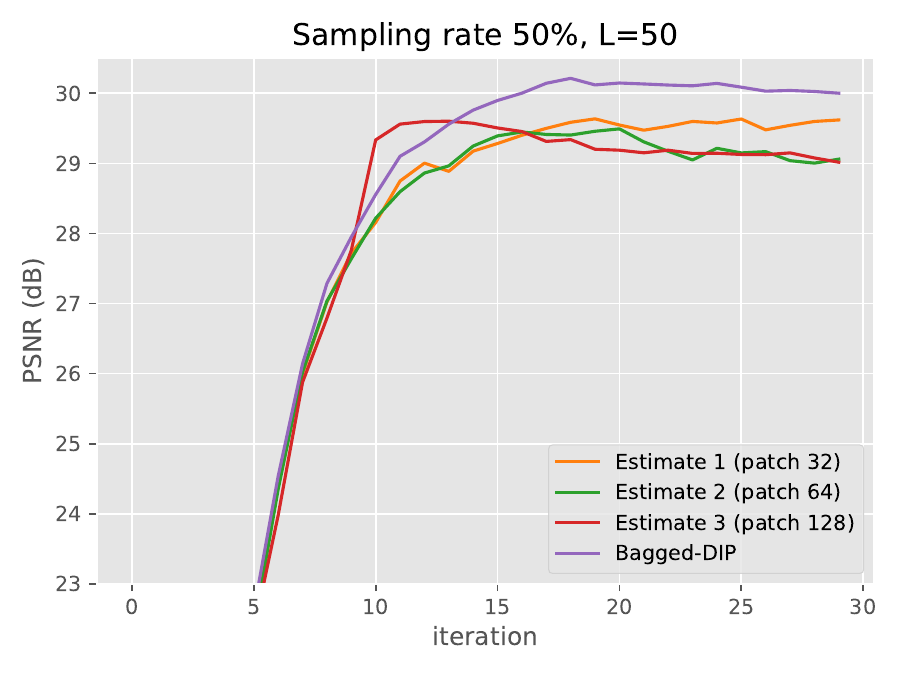}
\includegraphics[width=0.4\textwidth]{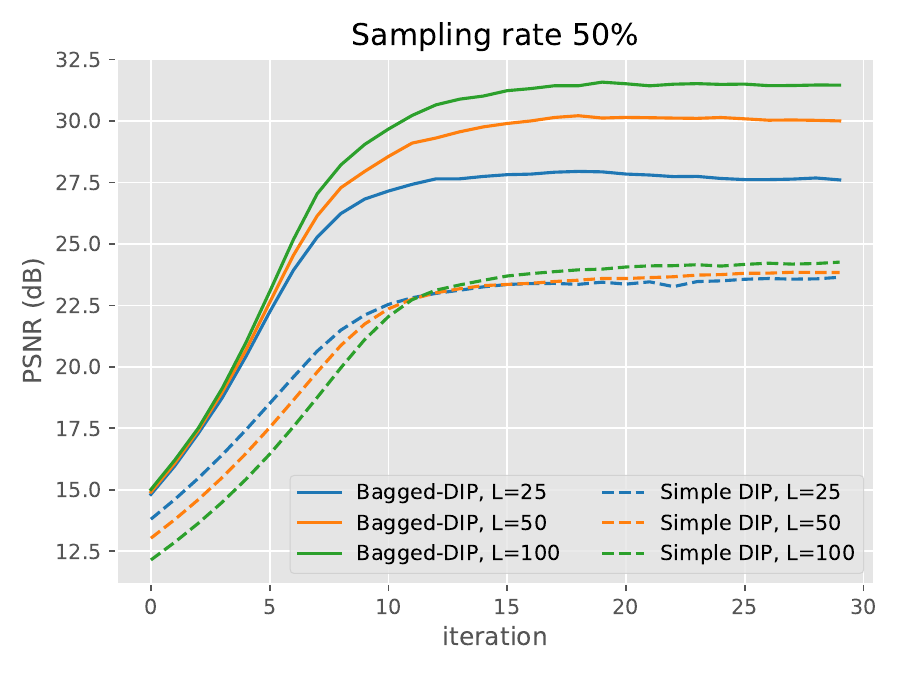}
\vspace{-.4cm}
\caption{(Left) We compare a Bagged-DIP with three sophisticated DIP estimates. (Right) We compare PGD with simple and Bagged-DIPs across different looks on image ``Cameraman".}
\label{fig:simple_bagged_L}
\end{figure}
\vspace{-.4cm}

\subsubsection{Simple architectures versus Bagged-DIPs}
So far our simulations have been focused on sophisticated networks. Are simpler networks that trade variance for the bias able to offer better performance? The right panel of Figure \ref{fig:simple_bagged_L} compares the performance of Bagged-DIP-based PGD with that of PGD with a simple DIP. Not only this figure shows the major improvement that is offered by using more complicated networks (in addition to bagging), but also it clarifies one of the serious limitations of the simple networks. Note that as $L$ increases, the performance of PGD with simple DIP is not improving. In such cases, the low-accuracy of DIP blocks the algorithm from taking advantage of extra information offered by the new looks. More results about the comparison between Bagged-DIP and simple structured DIP can be found in Table~\ref{tab:baselines_comparison} of Appendix~\ref{app:add_exp_baseline_DIP}.

\subsection{Performance of Bagged-DIP-based PGD on simulated data}
\label{sec:simulation_baggedDIP}
In this section, we offer a comprehensive simulaion study to evaluate the performance of the Bagged-DIP-based PGD on several images. We explore the following settings in our simulations:
\vspace{-.4cm}

\begin{itemize}
\item Number of looks ($L$): $L=25, 50, 100$.
\vspace{-.2cm}

\item Undersampling rate ($m\over n$): $\frac{m}{n}= 0.125, 0.25, 0.5$. 
\vspace{-.2cm}

\end{itemize}

For each combination of $L$ and $m/n$, we pick one of the $8$, $128 \times 128$ images mentioned in Table \ref{tab:main-CS-PSNR-SSIM}.\footnote{Images from the Set11~\cite{kulkarni2016reconnet} are chosen and cropped to $128 \times 128$ for computational manageability in Table~\ref{tab:main-CS-PSNR-SSIM}.} We then generate the matrix $A \in \mathbb{C}^{m \times n}$ by selecting the first $m$ rows of a matrix that is drawn from the Haar measure on the space of orthogonal matrices.  We then generate $\wv_1, \ldots, \wv_L \sim \mathcal{CN} (0,1)$, and for $l=1, 2, \ldots, L$, calculate $\yv_l = A X_o \wv_l$.

Regarding our implementations of Bagged-DIP-based PGD, we pick the following choices:
\vspace{-.4cm}

\begin{itemize}
\item Initialization: We initialize our algorithm with $\mathbf{x}_0 = \frac{1}{L} \sum^L_{l=1} |\Bar{A}^T \mathbf{y}_l|$. The final solution of DIP-based PGD seems not to depend on the initialization. (See Figure \ref{fig:init_compare} in Appendix \ref{sec:initialization:comp} for further information.)
\vspace{-.2cm}

\item Learning rate: We have chosen the learning rate of  $0.01$ for the gradient desent of the likelihood, and learning rate of  $0.001$ in the training of DIPs.
\vspace{-.2cm}

\item Number of iterations of SGD for training DIP: The details are presented in Table~\ref{tab:num_iteration_DIP} in the appendix.
\vspace{-.2cm}

\item Number of iterations of PGD: We run the outer loop (gradient descent of likelihood) for $100,200,300$ iterations when $m/n=0.5,0.25,0.125$ respectively.
\end{itemize}

The peak signal-to-noise-ratio (PSNR) and structural index similarity (SSIM) of our reconstructions are all reported in Table~\ref{tab:main-CS-PSNR-SSIM}. Qualitative results are presented in Figure~\ref{fig:qualitative_vis} of Appendix \ref{sec:qualitative_results}.

There are no other existing algorithms that can work in the undersampled regime $m<n$ we have considered in this paper. 
The sole algorithm addressing speckle noise in ill-conditioned and undersampled scenarios before our work is the vanilla PGD proposed in \cite{chen2023multilook}. Although designed for real-valued signals and measurements, we have adapted a complex-valued version of this algorithm, with results presented in Appendix~\ref{app:add_exp_baseline_DIP}. 
It can be seen that, for $L=100$, $L=50$, and $L=25$ on average (being averaged over $m/n=0.125,0.25,0.5$, and across all images) our algorithm outperforms the one presented in \cite{chen2023multilook} by $1.09$ dB, $1.47$ dB, and $1.27$ dB respectively.

\subsection{Sharpness of our theoretical results}

In this section, we would like to compare the results that we have here with the theoretical result we presented in Theorem \ref{thm:maintheorem}. Note that the dominant term in the MSE is the second term, i.e. $\frac{n \sqrt{ k \log n}}{  m \sqrt{L m}}$. There are two features of this term that we would like to confirm with our simulations:

\begin{enumerate}
\item As is clear, the decay of this term in terms of $m$ is $m^{3/2}$. Hence, if we double $m$, we expect to recieve the extra decay of $2^{3/2}$ in MSE, or the gain of $15 \log 2 \approx 4.5$dB in PSNR. The average of the gains we have received in average PSNR has been $3.99$ dB, which is within the error bounds from $4.5$dB that comes from our theoretical work. 
\vspace{-.2cm}

\item As is clear, the decay of this term in terms of $m$ is $L^{1/2}$. Hence, if we double $L$, we expect to recieve the extra decay of $\sqrt{2}$ in MSE, or the gain of $5 \log 2 \approx 1.5$dB in PSNR. The average of the gains we have received in average PSNR has been $1.42$ dB, which is within the error bounds from $1.5$dB that comes from our theoretical work. 
\vspace{-.2cm}

\end{enumerate}

\subsection{Comparison with $A= I$ case}
The goal of this section is to provide a performance comparison with cases where we have control over the matrix $A$, allowing us to design it as a well-conditioned kernel. Ideally, we can assume that the measurements are in the form of  $\yv_l = X \wv_l$, i.e. $A=I$. In this case, the task is transformed into classical despeckling problem. Since there is no wide or ill-conditioned matrix $A$ involved in the measurment process, we expect the imaging systems to outperform for instance the $50\%$ downsampled examples we presented in Table \ref{tab:main-CS-PSNR-SSIM}. Hence, the main question we aim to address here is:

\vspace{-.4cm}
\begin{itemize}
\item What is the PSNR cost incurred due to the undersampling of our measurement matrices? 
\end{itemize}
\vspace{-.2cm}

To address this question, we do the following empirical study: having access to the $L$ measurements in the form of $\yv_l = X \wv_l$, $l=1, 2, \ldots, L$, we create a sufficient statistic for estimating $X$. The sufficient statistics is the matrix $S = (\frac{1}{L} \sum^L_{l=1} \yv_l \yv^T_l)^{1/2}$. Then, this sufficient statistic is fed to DnCNN neural network~\cite{zhang2017beyond,zhang2017learning} (which we term it as DnCNN-UB), and learn the network to achieve the best denoising performance. The results of DnCNN-UB are often between 1-3dB better than the results of our Bagged-DIP based PGD. With the exception of the cameraman, where our Bagged-DIP based PGD seems to be better. Note that the 1-3dB gain is obtained for two reasons: (1) DnCNN approach uses a training data, while our Bagged-DIP-based approach does not require any training. (2) In DnCNN approach we have controlled the measurement matrix to be very well-conditioned. Details of DnCNN experimental settings and results can be found in Appendix~\ref{app:add_exp_DnCNN}.




\section{Proofs of the main results}\label{sec:theoryproof}
\subsection{Preliminaries}\label{ssec:prelim}
Before stating the proofs, we present a few lemmas that will be used later in the proof. 
\begin{lemma}\label{lem:boundeigenvalues}
Let $B$ and $C$ denote two $n \times n$ symmetric and invertible matrices. Then, if $\lambda_i$ represents the $i^{\rm th}$ eigenvalue of $B^{-1}- C^{-1}$, we have $|\lambda_i| \in [-\frac{\sigma_{\max} (B-C)}{\sigma_{\min}(B) \sigma_{\min}(C) }, \frac{\sigma_{\max} (B-C)}{\sigma_{\min}(B) \sigma_{\min}(C) }]$.
\end{lemma}
\begin{proof}
Suppose $\lambda_i$ is the $i^{\rm th}$ eigenvalue of $B^{-1}-C^{-1}$. Then, there exists a norm $1$ vector $\bm{v} \in \mathbb{R}^n$ such that 
\begin{equation*}
(B^{-1}-C^{-1}) \bm{v} = \lambda_i \bm{v}. 
\end{equation*}
Multiplying both sides by $B$, we have
\begin{equation*}
(I-BC^{-1}) \bm{v} = \lambda_i B \bm{v}. 
\end{equation*}
Define $\bm{u} = C^{-1}\bm{v}$. Then, we have
$(C-B) \bm{u} = \lambda_i BC \bm{u},$ or equivalently
\[
\lambda_i  \bm{u} = (BC)^{-1} (C-B) \bm{u}. 
\]
Hence,
\[
|\lambda_i| \leq \frac{\sigma_{\max}(C-B)}{\sigma_{\min} (B) \sigma_{\min} (C)}. 
\]
\end{proof}

\begin{lemma}\label{lem:singvalues}\cite{RudelsonVershinin2010}
Let the elements of an $m \times n$ ($m<n$) matrix $A$ be drawn independently from $\Nc(0,1)$. Then, for any $t>0$,
\begin{align}
\Prob(\sqrt{n}-\sqrt{m}- t \leq \sigma_{\min} (A) \leq \sigma_{\max}(A) &\leq \sqrt{n}+\sqrt{m}+ t) \geq 1-  2 {\rm e}^{-\frac{t^2}{2}}.
\end{align}
\end{lemma}

\begin{lemma}[Concentration of $\chi^2$ \cite{jalali2014minimum}]\label{lem:conc:chisq}
Let $Z_1, Z_2, \ldots, Z_n$ denote a sequence of independent $\Nc(0,1)$ random variables. Then, for any $t\in(0,1)$, we have
\[
\Prob (\sum_{i=1}^n Z_i^2 \leq n (1-t)) \leq {\rm e}^{\frac{n}{2} (t + \log (1- t)) }. 
\]
Also, for any $t>0$,
\[
\Prob (\sum_{i=1}^n Z_i^2 \geq n (1+t)) \leq {\rm e}^{-\frac{n}{2} (t -\log (1+ t)) }. 
\]
\end{lemma}

%

Define

\begin{theorem}[Hanson-Wright inequality] \label{thm:HW-ineq}
Let $\Xv= (X_1,...,X_n)$ be a random vector with independent components with $\E[X_i]=0$ and $\|X_i\|_{\Psi_2}\leq K$. Let A be an $n\times n$ matrix. Then, for $t>0$, 
\begin{align}
\Prob \Big(|\Xv^T A\Xv-\E[\Xv^T A\Xv]|>t \Big) \leq 2\exp\left(-c\min\left({t^2\over K^4\|A\|_{\rm HS}^2 },{t \over K^2\|A\|_2 }\right)\right),
\end{align}
where $c$ is a constant, and $\|X\|_{\psi_2} = \inf\{t>0 : \mathbb{E} (\exp (X^2/t^2)) \leq 2\}. 
$
\end{theorem}

\begin{theorem}[Decoupling of U-processes, Theorem 3.4.1. of \cite{de2012decoupling}]\label{thm:decoupling} Let $X_1, X_2, \ldots, X_n$ denote random variables with values in measurable space $(S, \mathcal{S})$. Let $(\tilde{X}_1, \tilde{X}_2, \ldots, \tilde{X}_n)$ denote an independent copy of $X_1, X_2, \ldots, X_n$. For $ i \neq j$ let $h_{i,j} : S^2 \rightarrow \mathbb{R}$. Then, there exists a constant $C$ such that for every $t >0$ we have
\[
\Prob \Big(|\sum_{i \neq j} h_{i,j} (X_i, X_j) | > t\Big) \leq C \Prob \Big(C |\sum_{i \neq j} h_{i,j} (X_i, \tilde{X}_j) | > t\Big). 
\] 

\end{theorem}

 \subsection{Main steps of the proof}\label{ssec:mainstepsappendix}
As we discussed in Section \ref{sec:main:theory}, the proof has three main steps:

\begin{enumerate}
\item Obtaining a lower bound for $\bar{f}(\tilde{\Sigma}_o)-\bar{f}(\Sigma_o)$ in terms of $\|\xvt_o- \xv_o\|_2^2$: This will be presented in Section \ref{ssec:lowerbd}.  

\item Finding upper bounds for $\bar{f}(\tilde{\Sigma}_o)- f(\tilde{\Sigma}_o)$ and $f (\Sigma_o) - \bar{f}({\Sigma}_o)$: This will be presented in Section \ref{ssec:upperconcentration}. 

\item Finding an upper bound for $f(\tilde{\Sigma}_o)-{f}(\hat{\Sigma}_o)$: this will be presented in Section \ref{ssec:errorsfunction}.   
\end{enumerate}

We combine these steps and finish the proof in Section \ref{ssec:proofsummary}. For notational simplicty we drop the subscript of $\delta_n$ in our proof. However, whenever we need the interpretation of the results we should not that $\delta$ depends on $n$ and behaves like $O(1/n^5)$. It becomes more clear why we have picked this choice later. Many other choices of $\delta_n$ will work as well.

\subsubsection{Obtaining lower bound for $\bar{f}(\tilde{\Sigma}_o)-\bar{f}(\Sigma_o)$}\label{ssec:lowerbd}

Define $\Delta \Sigma$ as 
\[
\Delta \Sigma = \tilde{\Sigma}_o-\Sigma_o,
\] and let $\lambda_{m}$ denote the maximum eigenvalue of $\Sigma_o^{-{1\over 2}}\Delta \Sigma\Sigma_o^{-{1\over 2}}$. As the first step the following lemma obtains a lower bound $\bar{f}(\tilde{\Sigma}_o)-\bar{f}(\Sigma_o)$ in terms of $\tilde{\Sigma} - \Sigma_o$. 

\begin{lemma}\label{lemma:1} \cite{zhou2022compressed}
For $\xvt \in\mathbb{R}^n$ and  $\xv_o\in\mathbb{R}^n$, let $\tilde{X}=\diag(\xvt)$, $X_o=\diag(\xv_o)$.  Assume that $A\tilde{X}^2A^T$ and $AX_o^2A^T$ are both invertible, and define $\tilde{\Sigma}=(A\tilde{X}^2A^T)^{-1}$, and $\Sigma_o=(AX_o^2A^T)^{-1}$.  Then,
\begin{align}\label{eq:thm1-step2_0}
\bar{f}(\tilde{\Sigma})-\bar{f}(\Sigma_o)&\geq {1\over 2(1+\lambda_{m})^2}{\rm Tr}(\Sigma_o^{-1}\Delta \Sigma\Sigma_o^{-1}\Delta \Sigma),
\end{align}
\end{lemma}
 The next step of the proof is to connect the quantity ${\rm Tr}(\Sigma_o^{-1}\Delta \Sigma\Sigma_o^{-1}\Delta \Sigma)$ appearing in  \eqref{eq:thm1-step2_0} to the difference $\xvt-\xv_o$. The subsequent lemma brings us one step closer to this goal.

\begin{lemma}\label{lemma:vector} \cite{zhou2022compressed}
Consider two $m\times m$ matrices $\tilde{\Sigma}=(A\tilde{X}^2A^T)^{-1}$ and $\Sigma=(AX^2A^T)^{-1}$ and define $\Delta \Sigma=\tilde{\Sigma}-\Sigma$.  Then,
\begin{eqnarray}\label{eq:lowerbound:2nd}
{\rm Tr}(\Sigma^{-1}\Delta \Sigma\Sigma^{-1}\Delta \Sigma) 
\geq {x^4_{\min} \lambda^2_{\min}(AA^T) \over x^8_{\max} \lambda^4_{\max}(AA^T)  }\|A(\tilde{X}^2-X^2)A^T\|^2_{\rm HS}
\end{eqnarray}
\begin{eqnarray}
{\rm Tr}(\Sigma^{-1}\Delta \Sigma\Sigma^{-1}\Delta \Sigma)\leq {x^4_{\max} \lambda^2_{\max}(AA^T) \over x_{\min}^8 \lambda^4_{\min}(AA^T)  }\|A(\tilde{X}^2-X^2)A^T\|^2_{\rm HS}.   
\end{eqnarray}
\end{lemma}
 Combining \eqref{eq:thm1-step2_0} and \eqref{eq:lowerbound:2nd} enables us to obtain a lower bound for $\bar{f}(\tilde{\Sigma}_o)-\bar{f}(\Sigma_o)$ in terms of $\|A(\tilde{X}_o^2-X_o^2)A^T\|^2_{\rm HS}$ and $\lambda_{\min} (AA^T)$ and $\lambda_{\max} (AA^T)$. Lower bounding the quantity ${ \lambda^2_{\min}(AA^T) \over \lambda^4_{\max}(AA^T)  }$ is straightforward. Define the event:
\[
\Ec_4=\{\sqrt{n}-2\sqrt{m}\leq \sigma_{\min} (A) \leq \sigma_{\max}(A) \leq \sqrt{n}+2\sqrt{m}\},
\]
From Lemma \ref{lem:singvalues},  we have
\begin{equation}\label{eq:probE4}
\Prob(\Ec_4^c)\leq 2\exp(-\frac{m}{2}).
\end{equation} 
Hence,
\begin{align}\label{eq:lowerAAT}
\Prob& \left({ \lambda^2_{\min}(AA^T) \over  \lambda^4_{\max}(AA^T)} \leq \frac{(\sqrt{n} - 2 \sqrt{m})^4}{(\sqrt{n} + 2 \sqrt{m})^8}\right)\leq 2\exp(-\frac{m}{2}). 
\end{align}

The only remaining step in obtaining the lower bound we are looking for is that, the term $\|A(\tilde{X}_o^2-X_o^2)A^T\|^2_{\rm HS}$  in the lower bound is not in the form of $\|\xvt_o-\xv_o\|_2^2$ yet. Hence, the final step in obtaining the lower bound is to obtain a lower bound of the form $\|\xvt_o - \xv_o\|_2^2$ for $\|A(\tilde{X}_o^2-X_o^2)A^T\|^2_{\rm HS}$. Towards this goal, for $\gamma>0$ define $\Ec_1 (\gamma)$ as the event that 
\[
 \|A(\tilde{X}_o^2-X_o^2)A^T\|^2_{\rm HS}\geq m(m-1) \|\xv_o^2-\xvt_o^2\|_2^2 - m^2 n \gamma,
\]
Our goal is to show that for an appropriate value of $\gamma$ this event holds with high probability. Towards this goal we use the following lemma:
\begin{lemma} \label{lem:lowerAX2AT}
Let the elements of $m\times n$ matrix $A$ be drawn  i.i.d.~ $\Nc(0,1)$. For any given $\bf{d} \in \mathbb{R}^n$, define $D = \diag (\bf{d})$. Then,
\begin{eqnarray}
\Prob (\|A DA^T\|^2_{HS} \leq m(m-1) \sum_{i=1}^n d_i^2 - t )
\leq 2C \exp \left( -c \min \Big( \frac{t^2}{C^2   \| \mathbf{d}\|_{\infty}^{{4}}  q_{m,n} }, \frac{t}{C \| \mathbf{d}\|_{\infty}^{2} \tilde{q}_{m,n}} \Big)\right) + 2 {\rm e}^{-\frac{n}{2}},
 \end{eqnarray}
 where $C$ and $c$ are the constants that appeared in Lemmas \ref{thm:decoupling} and \ref{thm:HW-ineq}, and 
 \begin{eqnarray}
 q_{m,n} &\triangleq& m^2(2 \sqrt{n} + \sqrt{m})^4, \nonumber \\
 \tilde{q}_{m,n} &\triangleq& (2 \sqrt{n} + \sqrt{m})^2. 
 \end{eqnarray} 
\end{lemma}

The proof of this lemma uses decoupling and is presented in Section \ref{sec:proof:lem:lowerAX2AT}. Using this lemma, we have
\begin{align}
\Prob& (\Ec_1^c) \overset{(a)}{\leq} 2C e^{k \log \frac{2k}{\delta}}   \exp \left( -c \min \Big( \frac{\check{\alpha}_{m,n} \gamma^2}{x_{\max}^8}, \frac{\check{\beta}_{m,n} \gamma}{x_{\max}^4} \Big)\right)+ 2 {\rm e}^{k \log \frac{2k}{\delta} - n/2} \nonumber \\
&\leq 2C e^{k \log \frac{2k}{\delta}} \left(  {\rm e}^{-c \frac {\check{\alpha}_{m,n} \gamma^2}{x_{\max}^8}}+  {\rm e}^{ -c  \frac{\check{\beta}_{m,n} \gamma}{ x_{\max}^4}}  \right) + 2 {\rm e}^{k \log \frac{2k}{\delta} - \frac{n}{2}} \nonumber \\
&=2C {\rm e}^{k \log \frac{2k}{\delta}}   {\rm e}^{ -  \frac{cm^2 \gamma^2}{C^2   x^8_{\max}  (2  + \sqrt{m/n})^4 }} +2C {\rm e}^{k \log \frac{2k}{\delta}}   {\rm e}^{ - \frac{cm^2  \gamma}{C x_{\max}^4 (2 + \sqrt{m/n})^2} } + 2 {\rm e}^{k \log \frac{2k}{\delta} - n/2},
 \end{align}
 where 
 \begin{eqnarray}
 \check{\alpha}_{m,n} &\triangleq& \frac{m^4n^2 }{C^2  m^2  (2 \sqrt{n} + \sqrt{m})^4 } = \frac{m^2n^2 }{C^2   (2 \sqrt{n} + \sqrt{m})^4 }, \nonumber \\
 \check{\beta}_{m,n}  &\triangleq&  \frac{m^2 n }{C  (2 \sqrt{n} + \sqrt{m})^2}. 
\end{eqnarray}
In order to obtain Inequality (a), we should first note that by a simple counting argument we conclude that  $|\mathcal{C}_{\delta}| \leq (\frac{2k}{\delta})^k$ \cite{shalev2014understanding}. Hence, by combining this result, the union bound on the choice of $\tilde{x}_o$ and Lemma \ref{lem:lowerAX2AT} we reach Inequality (a).

 By setting 
\begin{equation}
\gamma= 2C \frac{x_{\max}^4 (2+ \sqrt{m/n})^2 }{m \sqrt{c}} \sqrt{k \log \frac{2k}{\delta}}, 
\end{equation}
we have
\begin{align}\label{eq:lower:axat}
\|A (\tilde{X}_o ^2 - X_o^2 )A^T\|_{HS}^2 &\geq m(m-1) \sum_{i}^n (\tilde{x}_{o,i}^2 - x_{o,i}^2)^2\nonumber - \tilde{C} mn \sqrt{k \log \frac{2k}{\delta}} \nonumber \\
&= m(m-1) \sum_{i}^n (\tilde{x}_{o,i} - x_{o,i})^2 (\tilde{x}_{o,i} + x_{o,i})^2 - \tilde{C} mn \sqrt{k \log \frac{2k}{\delta}} \nonumber \\
&\geq 4m(m-1) x_{\min}^2 \sum_{i}^n (\tilde{x}_{o,i} - x_{o,i})^2 \nonumber - \tilde{C} mn \sqrt{k \log \frac{2k}{\delta}} \nonumber \\
&= 4m(m-1) x_{\min}^2 \|\xvt_o- \xv_o\|_2^2  - \tilde{C} mn \sqrt{k \log \frac{2k}{\delta}}
\end{align}
with probability
\begin{eqnarray}\label{eq:probE1}
\Prob (\Ec_1^c) &\leq& O(e^{-k \log \frac{k}{\delta}}   +  {\rm e}^{k \log \frac{k}{\delta}-\frac{n}{2}}), 
\end{eqnarray}
In the above equations $\tilde{C}$ is a constant that does not depend on $m,n$ or $\delta$. Furthermore, in \eqref{eq:probE1} we have assumed that $m$ is large enough (and hence $\gamma$ is small enough) to make the inequality $\frac{m^2 \gamma^2}{C^2   x^8_{\max}  (2  + \sqrt{m/n})^4 }< \frac{m^2  \gamma}{C x_{\max}^4 (2 + \sqrt{m/n})^2}$ true.

Combining \eqref{eq:thm1-step2_0}, \eqref{eq:lowerbound:2nd}, \eqref{eq:lowerAAT}, \eqref{eq:lower:axat}, and  \eqref{eq:probE1} we conclude the lower bound we were looking for:

\begin{align}\label{eq:lower:semifinal}
\bar{f}(\tilde{\Sigma}_o)-\bar{f}(\Sigma_o) \geq \frac{\check{C}}{(1+ \lambda_m)^2} \frac{(\sqrt{n} - 2\sqrt{m})^4}{(\sqrt{n} + 2\sqrt{m})^8}  \left(4m(m-1) x_{\min}^2 \|\xvt_o- \xv_o\|_2^2  - \tilde{C} mn \sqrt{k \log \frac{2k}{\delta}} \right)
\end{align}
with probability 
\[
P(\mathcal{E}_1^c \cap \mathcal{E}_4^c) \geq 1  - O\left(e^{-\frac{m}{2}}+ e^{-k \log \frac{k}{\delta}}   +  {\rm e}^{k \log \frac{k}{\delta}-\frac{n}{2}}\right).
\]

Our last step for the lower bound is to simplify the expression of \eqref{eq:lower:semifinal}.  Recall that $\lambda_{m}$ is defined as the maximum eigenvalue of $\Sigma_o^{-{1\over 2}}\Delta \Sigma\Sigma_o^{-{1\over 2}}$. On the other hand, $\lambda_m=\|\Sigma_o^{-{1\over 2}}\Delta \Sigma\Sigma_o^{-{1\over 2}}\|_2 \leq \|\Delta \Sigma\|_2 \|\Sigma_o^{-{1\over 2}}\|_2^2 = \|\Delta \Sigma\|_2 \|\Sigma_o^{-1}\|_2$. But, $ \|\Sigma_o^{-1}\|_2=\|AX_o^2A^T\|_2\leq x^2_{\max}\lambda_{\max}(AA^T)$. Similarly, $\|\Delta \Sigma\|_2 =\|\hat{\Sigma}_o-{\Sigma}_o\|_2 \leq \|\hat{\Sigma}_o\|_2 + \|{\Sigma}_o\|_2 \leq {1\over \|AX_o^2A^T\|_2}+{1\over \|A\hat{X}_o^2A^T\|_2}\leq {2\over x^2_{\min}\lambda_{\min}(AA^T)}$. So overall, $\lambda_m\leq {2x^2_{\max}\lambda_{\max}(AA^T)\over x^2_{\min}\lambda_{\min}(AA^T)}$, and conditioned on $\Ec_4$, we have
\begin{align}
\lambda_m\leq {2x^2_{\max}(1+2\sqrt{m/n})^2\over x^2_{\min}(1-2\sqrt{m/n})^2}.\label{eq:def-lambda-m}
\end{align}
Hence, since $\lambda_m>1$, we have
\begin{align}\label{eq:bound:lambdam}
(1+\lambda_m)^2\leq {16x^4_{\max}(1+2\sqrt{m/n})^4\over x^4_{\min}(1-2\sqrt{m/n})^4}= O(1). 
\end{align}
Hence,
\begin{eqnarray}
\frac{\check{C}}{(1+ \lambda_m)^2} \frac{(\sqrt{n} - 2\sqrt{m})^4}{(\sqrt{n} + 2\sqrt{m})^8}
= \frac{\check{C}}{n^2(1+ \lambda_m)^2} \frac{(1 - 2\sqrt{m/n})^4}{(1 + 2\sqrt{m/n})^8} \geq \frac{\check{\check C}}{n^2}. 
\end{eqnarray}
In summary, there exist two absolute constants $\bar{C}$ and $\bar{\bar{C}}$ such that
\begin{align}\label{eq:lower:final}
\bar{f}(\tilde{\Sigma}_o)-\bar{f}(\Sigma_o)
\geq   \frac{ \bar{C} m(m-1) }{n^2}\|\xvt_o- \xv_o\|_2^2  - \bar{\bar{C}} \frac{m \sqrt{k \log \frac{2}{\delta}}}{n}
\end{align}
with probability 
\[
P(\mathcal{E}_1^c \cap \mathcal{E}_4^c) \geq 1  - O\left(e^{-\frac{m}{2}}+ e^{-k \log \frac{1}{\delta}}   +  {\rm e}^{k \log \frac{k}{\delta}-\frac{n}{2}}\right).
\]

\subsubsection{Finding upper bounds for $\bar{f}(\tilde{\Sigma}_o)- f(\tilde{\Sigma}_o)$ and $f (\Sigma_o) - \bar{f}({\Sigma}_o)$}\label{ssec:upperconcentration}

Given $t_2>0$ define events $\Ec_2$ and $\Ec_3$ 
\[
\Ec_2(t_2)=\{ |\delta f({\Sigma}_o)| \leq t_2\},  \;\;\;\Ec_3=\{ |\delta f(\tilde{\Sigma}_o)|\leq t_2\}.
\]

The following Lemma enables to calculate the probability of $\Ec_2 \cap \Ec_3$. 

\begin{lemma}\label{lemma:concent-z}
Given $\Sigma=(AX^2A^T)^{-1}$, let $\delta f({\Sigma})={f}({\Sigma})-\bar{f}({\Sigma})$. Then, for $t>0$, there exists a constant $c$ independent of $m,n, x_{\rm min},$ and $x_{\rm max}$, such that
\[
\Prob(|\delta f(\Sigma)|\geq t|A)\leq  2\exp\left( -\frac{cLt}{2x_{\max}^2 \| A^T\Sigma A\|_2 } \min \left(1, \frac{t}{2m x_{\max}^2 \| A^T\Sigma A\|_2} \right) \right).
\]
Also, $\|A^T\Sigma A\|^2\leq {\lambda^2_{\max}(AA^T)\over \lambda^2_{\min}(AA^T) x^4_{\min}}$. 
\end{lemma}
The proof of this lemma is presented in Section \ref{sec:proof:lemma:concent-z}. 

Our goal is to use Lemma \ref{lemma:concent-z} for obtaining  $\Prob((\Ec_2\cap\Ec_3)^c)$. 

First note that we have
\begin{eqnarray}\label{eq:breakingProb}
\Prob((\Ec_2\cap\Ec_3)^c) = \Prob((\Ec_2\cap\Ec_3)^c\cap\Ec_4)
+ \Prob((\Ec_2\cap\Ec_3)^c\cap\Ec_4^c) \leq \Prob((\Ec_2\cap\Ec_3)^c\cap\Ec_4) +\Prob(\Ec_4^c).
\end{eqnarray}

Furthermore, using Lemma \ref{lemma:concent-z}, for $t_2 \leq 2m x_{\max}^2 \| A^T\Sigma A\|_2 $,
\begin{align}
\Prob&((\Ec_2\cap\Ec_3)^c\cap\Ec_4)\leq  e^{k \log \frac{2k}{\delta}}\exp(-{c_3L t_2^2/2m})\nonumber\\
&=\exp((k \log \frac{2k}{\delta}) -{c_3L t_2^2/2m}).
\end{align}
To obtain the first inequality we have used Lemma \ref{lemma:concent-z} and used a union bound on all the possible choices of $\tilde{X}_o \in \Cc_n$. As discussed before, $|\Cc_\delta| \leq {\rm e}^{k \log \frac{2k}{\delta}}$. 
Let $t_2=\sqrt{4m k \log( \frac{2k}{\delta}) /Lc_3}$. Note that since we are interested in the regime that $m$ is much bigger than $k \log \frac{2k}{\delta}$, we can conclude that $t_2 \leq 2m x_{\max}^2 \| A^T\Sigma A\|_2 $. Hence, 
\begin{align}\label{eq:probcomp}
\Prob((\Ec_2\cap\Ec_3)^c\cap\Ec_4)&\leq  e^{-k \log \frac{2k}{\delta}}.
\end{align}
Hence, combining \eqref{eq:breakingProb}  and \eqref{eq:probcomp} we have
\begin{equation}
\Prob((\Ec_2\cap\Ec_3)^c) \leq e^{-k \log \frac{2k}{\delta}}+ 2e^{-\frac{m}{2}}.
\end{equation}

In summary, we have that for $t_2 = \sqrt{4m k \log( \frac{2k}{\delta}) /Lc_3}$, 

\begin{equation}\label{eq:probE2E3}
\Prob((\Ec_2\cap\Ec_3)^c) = O (e^{-k \log \frac{k}{\delta}}+ e^{-\frac{m}{2}}). 
\end{equation}

\subsubsection{Finding an upper bound for $f(\tilde{\Sigma}_o)-{f}(\hat{\Sigma}_o)$:}\label{ssec:errorsfunction}

The following lemma help us obtain an upper bound for $f(\tilde{\Sigma}_o)-{f}(\hat{\Sigma}_o)$.\

\begin{lemma}\label{lemma:1:upper}
 Assume that $A\tilde{X}_o^2A^T$ and $A\hat{X}_o^2A^T$ are both invertible. Define $\Delta \Sigma=\tilde{\Sigma}_o-\hat{\Sigma}_o$. Then,
\begin{equation}
|\lambda_i (\tilde{\Sigma}_o^{-{1\over 2}}\Delta \Sigma \tilde{\Sigma}_o^{-{1\over 2}}) | \in [0, \frac{x_{\max}^2\lambda^2_{\max} (A A^T) \| {\xvh}_o^2 - {\xvt_o}^2\|_\infty}{x_{\min}^4 \lambda^2_{\min} (AA^T)} ].
\end{equation}
Furthermore, if all the eigenvalues of $\tilde{\Sigma}_o^{-{1\over 2}}\Delta \Sigma\tilde{\Sigma}_o^{-{1\over 2}}$ fall in the range $[-0.5, 0.5]$. Then,
\begin{eqnarray}\label{eq:upperconvexity}
{f}(\tilde{\Sigma}_o)-{f}(\hat{\Sigma}_o)\leq  \frac{x_{\max}^2 \lambda^2_{\max}(AA^T)  \|{\xvh_o}^2 - {\xvt}_o^2\|_\infty}{x^4_{\min} \lambda_{\min}^2 (AA^T)} \Big(2m + \frac{1}{L \sigma_w^2} \sum_{\ell=1}^L \wv_{\ell}^T \wv_{\ell}\Big).  
\end{eqnarray}
\end{lemma}

The proof of this lemma can be found in Section \ref{sec:proof:lemma:1:upper}. The main objective of this section is to obtain upper bound for the following three terms:
\begin{itemize}
\item $\frac{1}{L \sigma_w^2} \sum_{\ell} \wv_{\ell}^T \wv_{\ell}$: 

Consider the following event:
\[
\Ec_5 = \left\{ \frac{1}{L \sigma_w^2} \sum_{\ell=1}^{L} \wv_{\ell}^T \wv_{\ell} \geq 2n \right\}
\]
It is straightforward to use Lemma \ref{lem:conc:chisq} to see that
\[
\Prob(\Ec_5) \leq {\rm e}^{-\frac{Ln}{8}}. 
\]
Hence we conclude that
\begin{eqnarray}
 \frac{1}{L \sigma_w^2} \sum_{\ell=1}^{L} \wv_{\ell}^T \wv_{\ell} \leq 2n. 
\end{eqnarray}
with probability
\begin{equation}
    \Prob(\Ec_5^c) \geq 1 - {\rm e}^{-\frac{Ln}{8}}.
\end{equation}.

\item $\frac{ \lambda^2_{\max}(AA^T)  }{ \lambda_{\min}^2 (AA^T)} $:

Based on the calculation of $\Prob(\Ec_4)$ in \eqref{eq:probE4}, it is straightforward to see that
\begin{align}
&\frac{ \lambda^2_{\max}(AA^T)  }{ \lambda_{\min}^2 (AA^T)}  \leq \frac{(\sqrt{n} + 2\sqrt{m})^2}{(\sqrt{n} - 2\sqrt{m})^2}, \nonumber
\end{align}
with probability larger than
\begin{equation}
\Prob(\Ec_4) \geq 1 - 2\exp(-\frac{m}{2}).
\end{equation} 

\item $\|{\xvh_o}^2 - {\xvt}_o^2\|_\infty$: Before we start let us prove a simple claim. For every $\xvt$ in the range of $g_{\thetav} (\uv)$, there exists a vector $\mathbf{v}$ such that,
\[
\|\xvt- \mathbf{v}\|_2 \leq  \delta. 
\]
To prove this claim, let's assume that $\xvt = g_{\tilde{\thetav}} (\uv)$. Suppose that $\hat{\thetav}$ is a vector in the $\delta_n$-net of $[-1,1]^k$that is closest to $\tilde{\thetav}$. By the definition of $\delta_n$-net, we have 
\[
\|\tilde{\thetav} - \hat{\thetav} \|_2 \leq \delta_n.
\]
Hence, it is straightforward to use Lipschitzness of $g_{\tilde{\thetav}} (\uv)$ to prove that
\begin{equation}\label{eq:xvt_xvh}
\|\xvt-\hat{\xv}\|_2 \leq \delta, 
\end{equation}
and
\[
\|\xvh_o^2- \xvt_o^2\|_\infty \leq 2 x_{\max} \|\xvh_o- \xvt_o\|_\infty \leq 2 x_{\max} \|\xvh_o- \xvt_o\|_2 \leq  2 x_{\max} \delta,
\]

\end{itemize}

Summarizing the discusions of this section, we can conclude that there exist a constant $\tilde{C}$ such that
\begin{align}
|f(\tilde{\Sigma}_o)-{f}(\hat{\Sigma}_o) | \leq \tilde{C} n \delta 
\end{align}
with probability larger than
\[
\Prob (\Ec_4 \cap \Ec_5) \geq 1 - O(e^{-\frac{m}{2}} + e^{-\frac{Ln}{8}}). 
\]

\subsubsection{Summary of the bounds}\label{ssec:proofsummary}

As mentioned in \eqref{eq:first_ineq} we have
\begin{align}\label{eq:mainf:equation}
\bar{f}(\tilde{\Sigma}_o)-\bar{f}(\Sigma_o)\leq \bar{f}(\tilde{\Sigma}_o)- f(\tilde{\Sigma}_o)+f(\tilde{\Sigma}_o)-{f}(\hat{\Sigma}_o)+f({\Sigma}_o)-\bar{f}({\Sigma}_o),
\end{align}
Furthermore, we proved the following:

\begin{itemize}
\item According to \eqref{eq:lower:final} there exist two constants $\bar{C}$ and $\bar{\bar{C}}$ such that
\begin{align}\label{eq:lb:2}
\bar{f}(\tilde{\Sigma}_o)-\bar{f}(\Sigma_o) \geq   \frac{ \bar{C} m(m-1) }{n^2}\|\xvt_o- \xv_o\|_2^2  - \bar{\bar{C}} \frac{m \sqrt{k \log \frac{2k}{\delta}}}{n}
\end{align}
with probability 
\[
P(\mathcal{E}_1^c \cap \mathcal{E}_4^c) \geq 1  - O\left(e^{-\frac{m}{2}}+ e^{-k \log \frac{k}{\delta}}   +  {\rm e}^{k \log \frac{k}{\delta}-\frac{n}{2}}\right).
\]
\item According to the discussion of Section \ref{ssec:upperconcentration}, we have 
\begin{align}
|\bar{f}(\tilde{\Sigma}_o)- f(\tilde{\Sigma}_o)|  &\leq \sqrt{4m k \log( \frac{2k}{\delta}) /Lc_3}, \nonumber \\
|\bar{f} ({\Sigma}_o) - f (\Sigma_o)| &\leq \sqrt{4m k \log( \frac{2k}{\delta}) /Lc_3}, 
\end{align}
with probability $1- O (e^{-k \log \frac{k}{\delta}}+ e^{-\frac{m}{2}})$. 

\item According to the discussion of Section \ref{ssec:errorsfunction}, there exists a constant $\tilde{C}$ such that 
\begin{align}
|f(\tilde{\Sigma}_o)-{f}(\hat{\Sigma}_o) | \leq \tilde{C} n \delta 
\end{align}
with probability larger than
\[
\Prob (\Ec_4 \cap \Ec_5) \geq 1 - O(e^{-\frac{m}{2}} + e^{-\frac{Ln}{8}}). 
\]
\end{itemize}

Combining all these three results with \eqref{eq:mainf:equation} we notice that
\begin{align}
\frac{ \bar{C} m(m-1) }{n^2}\|\xvt_o- \xv_o\|_2^2 \leq \bar{\bar{C}} \frac{m \sqrt{k \log \frac{2}{\delta}}}{n} + 2 \sqrt{4m k \log( \frac{2}{\delta}) /Lc_3} + \tilde{C} n \delta ,
\end{align}
with probability $1- O(e^{-\frac{m}{2}} + e^{-\frac{Ln}{8}} + e^{-k \log \frac{k}{\delta}}+   {\rm e}^{k \log \frac{k}{\delta}-\frac{n}{2}})$. 
Hence, we can conclude that
\begin{align}
\frac{1}{n}\|\xvt_o- \xv_o\|_2^2 \leq O\left( \frac{ \sqrt{k \log \frac{k}{\delta} }}{m} +  \frac{n \sqrt{ k \log \frac{k}{\delta} }}{  m \sqrt{L m}} + \frac{ n^2 \delta } {m^2} \right),
\end{align}
with probability $1- O(e^{-\frac{m}{2}} + e^{-\frac{Ln}{8}} + e^{-k \log \frac{1}{\delta}}+   {\rm e}^{k \log \frac{1}{\delta}-\frac{n}{2}})$.

Set $\delta = \frac{1}{n^5}$, we have
\begin{align}\label{eq:thm:MSEbound2}
  \frac{1}{n} \|\xvt_o-\xv_o\|_2^2 =  O\left( \frac{ \sqrt{k \log n }}{m} +  \frac{n \sqrt{ k \log n}}{  m \sqrt{L m}}  \right),
\end{align}
with probability $1- O(e^{-\frac{m}{2}} + e^{-\frac{Ln}{8}} + e^{-k \log n}+   {\rm e}^{k \log n-\frac{n}{2}})$.

So far, we have found an upper bound for the error between $\|\tilde{\xv}_o-\xv_o\|_2$. However, our final estimate is $\hat{\xv}_o$. Note that we have
\begin{equation}
\|\hat{\xv}_o - \xv_o\|_2 \leq \|\hat{\xv}_o - \tilde{\xv}_o\|_2 + \|\tilde{\xv}_o- \xv_o\|_2 \leq  \delta + \|\tilde{x}_o -x_o\|_2.
\end{equation}
Hence, we have
\begin{align}\label{eq:thm:MSEbound_last}
  \frac{1}{n} \|\xvh_o-\xv_o\|_2^2 =  O\left( \frac{ \sqrt{k \log n }}{m} +  \frac{n \sqrt{ k \log n}}{  m \sqrt{L m}}  \right),
\end{align}
with probability $1- O(e^{-\frac{m}{2}} + e^{-\frac{Ln}{8}} + e^{-k \log n}+   {\rm e}^{k \log n-\frac{n}{2}})$.

\subsection{Proof of auxiliary lemmas}



\subsubsection{Proof of Lemma \ref{lem:lowerAX2AT} }\label{sec:proof:lem:lowerAX2AT}

Let $\bm{a}_i^T$ denote the $i^{\rm th}$ row of matrix $A$. We have
\begin{equation}
\|A DA^T\|^2_{HS} = \sum_{i=1}^m \sum_{j=1}^m |\bm{a}_i^T D {\bm a}_j|^2 \geq \sum_{i=1}^m \sum_{j \neq i} |{\bm a}_i^T D {\bm a}_j|^2. 
\end{equation}
Note that
\[
\E (\sum_{i=1}^m \sum_{j \neq i} |{\bm a}_i^T D {\bm a}_j|^2) = m(m-1) \sum_{i=1}^n d_i^2.  
\]
Using Theorem \ref{thm:decoupling} we conclude that there exists a constant $C$ such that 
\begin{align}\label{eq:lowerbounddec1}
&{\Prob (| \sum_{i=1}^m \sum_{j \neq i} |{\bm a}_i^T D {\bm a}_j|^2 -m(m-1) \sum_{i=1}^n d_i^2|> t )}\nonumber \\ 
&\leq C \Prob (C | \sum_{i=1}^m \sum_{j \neq i} |{\bm a}_i^T D \bm{\tilde{a}}_j|^2 -m(m-1) \sum_{i=1}^n d_i^2|> t ) \nonumber \\
&=C \Prob (C | \sum_{i=1}^m  {\bm a}_i^T D \sum_{j \neq i} \bm{\tilde{a}}_j  \bm{\tilde{a}}_j^T D {\bm a}_i -m(m-1) \sum_{i=1}^n d_i^2|> t ), 
\end{align}
where $\bm{\tilde{a}}_1, \bm{\tilde{a}}_2, \ldots, \bm{\tilde{a}}_m$ denote independent copies of $\bm{a}_1, \bm{a}_2, \ldots, \bm{a}_m$. Define $\tilde{A}$ as the matrix whose rows are $\bm{\tilde{a}}_1^T, \bm{\tilde{a}}_2^T, \ldots, \bm{\tilde{a}}_m^T$ . Also, let $\tilde{A}_{\backslash i}$ denote the matrix that is constructed by removing the $i^{\rm th}$ row of $\tilde{A}$. Define
\[
F\triangleq\left[\begin{array}{cccc} D \tilde{A}_{\backslash 1}^T \tilde{A}_{\backslash 1} D&0&\ldots&0 \\ 0 & D\tilde{A}_{\backslash 2}^T \tilde{A}_{\backslash 2}D&\ldots&0 \\ 0 &0&\ldots& D\tilde{A}_{\backslash m}^T {\tilde{A}}_{\backslash m}D \end{array} \right].
\]
and 
\[
\bm{v}^T = [\bm{a}_1^T, \bm{a}_2^T, \ldots, \bm{a}_m^T].  
\]
Using Theorem \ref{thm:HW-ineq} we have 
\begin{align}\label{eq:lowerbounddec2}
&{ \Prob (C | \sum_{i=1}^m  {\bm a}_i^T D \sum_{j \neq i} \bm{\tilde{a}}_j  \bm{\tilde{a}}_j^T D {\bm a}_i -m(m-1) \sum_{i=1}^n d_i^2|> t   \ | \ \tilde{A} )} \nonumber \\
&=  \Prob (C |\bm{v}^T F \bm{v} - \mathbb{E}\bm{v}^T F \bm{v}  |> t \ | \ \tilde{A}) \nonumber \\
&\leq 2 \exp \left( -c \min ( \frac{t^2}{C^2 \|F\|_{HS}^2 }, \frac{t}{C \|F\|_2 }) \right)
\end{align}
Hence, in order to obtain a more explicit upper bound, we have to find upper bounds for $\|F\|_2$ and $\|F\|_{HS}^2$. First note that 
\begin{align}\label{eq:uppedspecnormF}
\lambda_{\max}(F) &= \max_i (\lambda_{\max}(D\tilde{A}_{\backslash i}^T \tilde{A}_{\backslash i} D) )\nonumber \\
& \leq \lambda_{\max}(D\tilde{A}^T \tilde{A} D) \leq \| \mathbf{d}\|_{\infty}^2  \lambda_{\max}(\tilde{A}^T \tilde{A}).
\end{align}
Similarly,
\begin{eqnarray}
\|F\|_{HS}^2 &=& \sum_{i=1}^m \|D \tilde{A}_{\backslash i}^T \tilde{A}_{\backslash i} D\|_{HS}^2 \nonumber \\
& \overset{(a)}{\leq}& \sum_{i=1}^m  m \lambda_{\max}^2 (D\tilde{A}_{\backslash i}^T \tilde{A}_{\backslash i} D) \nonumber \\
& \overset{(b)}{\leq}& m^2 \| \mathbf{d}\|_{\infty}^{{4}}  \lambda^2_{\max}(\tilde{A}^T \tilde{A}),
\end{eqnarray}
where Inequality (a) uses the fact that the rank of matrix $D \tilde{A}_{\backslash i}^T \tilde{A}_{\backslash i} D$ is $m-1$, and Inequality (b) uses \eqref{eq:uppedspecnormF}. Finally, using Lemma \ref{lem:singvalues} we have
\begin{equation}
\Prob(\sigma_{\max}(\tilde{A}) > 2\sqrt{n}+\sqrt{m}) \leq 2 {\rm e}^{-\frac{n}{2}},
\end{equation}
and hence
\begin{equation}
\Prob(\lambda_{\max}(\tilde{A}^T A) > (2\sqrt{n}+\sqrt{m})^2) \leq 2 {\rm e}^{-\frac{n}{2}}.
\end{equation}

By combining \eqref{eq:lowerbounddec1} and \eqref{eq:lowerbounddec2} we obtain
\begin{align}\label{eq:lowerboundADA}
&\Prob (| \sum_{i=1}^m \sum_{j \neq i} |{\bm a}_i^T D {\bm a}_j|^2 -m(m-1) \sum_{i=1}^n d_i^2|> t \ | \ \tilde{A})  \nonumber \\
 &\leq 2C \mathbb{E} \left(\exp \left( -c \min \Big( \frac{t^2}{C^2  \|F\|_{HS}^2 }, \frac{t}{C \|F\|_2 } \Big)\right)\right),
\end{align}
where the expected value is with respect to the randomness in $F$ or equivalently $\tilde{A}$. 

Let the event $\mathcal{E}$ denote the event of $\sigma_{\max}(\tilde{A}) \leq 2\sqrt{n}+\sqrt{m}$, and $\mathbb{I}_{\mathcal{E}}$ denote the indicator function of the event $\mathcal{E}$. Then, using \eqref{eq:lowerboundADA} we have
\begin{align}\label{eq:finalADAT}
&{\Prob (| \sum_{i=1}^m \sum_{j \neq i} |{\bm a}_i^T D {\bm a}_j|^2 -m(m-1) \sum_{i=1}^n d_i^2|> t )} \nonumber \\
&= \Prob (\{ | \sum_{i=1}^m \sum_{j \neq i} |{\bm a}_i^T D {\bm a}_j|^2 -m(m-1) \sum_{i=1}^n d_i^2|> t \} \cap \mathcal{E})\nonumber \\
&+ \Prob (\{ | \sum_{i=1}^m \sum_{j \neq i} |{\bm a}_i^T D {\bm a}_j|^2 -m(m-1) \sum_{i=1}^n d_i^2|> t \} \cap \mathcal{E}^c)\nonumber \\
&\leq \mathbb{E} \left( \Prob (| \sum_{i=1}^m \sum_{j \neq i} |{\bm a}_i^T D {\bm a}_j|^2 -m(m-1) \sum_{i=1}^n d_i^2|> t \ | \ \tilde{A}) \mathbb{I}_{\mathcal{E}} \right) + \mathbb{P} (\mathcal{E}^c)  \nonumber \\
&\leq 2C \mathbb{E} \left(\exp \left( -c \min \Big( \frac{t^2}{C^2  \|F\|_{HS}^2 }, \frac{t}{C \|F\|_2 }  \Big)\right) \mathbb{I}_{\mathcal{E}}\right) + \mathbb{P} (\mathcal{E}^c)  \nonumber \\
&\leq 2C \exp \left( -c \min \Big( \frac{t^2}{C^2  \| \mathbf{d}\|_{\infty}^4  q_{m,n} }, \frac{t}{C \| \mathbf{d}\|_{\infty}^2 \tilde{q}_{m,n}} \Big)\right) + 2 {\rm e}^{-\frac{n}{2}}. 
\end{align}


\subsubsection{Proof of Lemma \ref{lemma:concent-z}} \label{sec:proof:lemma:concent-z}
By definition, 
\begin{align}
&\delta f({\Sigma})={f}({\Sigma})-\bar{f}({\Sigma}) \nonumber \\
&= {1\over L \sigma_w^2} \sum_{\ell=1}^L \wv_{\ell}^TX_oA^T\Sigma AX_o\wv_{\ell}-{\rm Tr}(\Sigma AX_o^2A^T).
\end{align}
Define matrix $B\in\mathbb{R}^{Ln\times Ln}$ as  $B=X_oA^T\Sigma AX_o$ and $\tilde{B}$ as
\[
\tilde{B} = \begin{bmatrix}
  B & 0 & \cdots & 0 \\
  0 & B & \cdots & 0 \\
  0 & 0 & \cdots & 0 \\
  0 & 0 & \cdots & B \\
\end{bmatrix}.
\]
Furthermore, define
\[
\tilde{\wv}^\top = [\wv_1, \wv_2, \ldots, \wv_L]. 
\]
Then, by the Hanson-Wright inequality (Theorem \ref{thm:HW-ineq}), we have 
\begin{align}
&\Prob(|{1\over L \sigma_w^2}\tilde{\wv}^T\tilde{B}\tilde{\wv}- {\rm Tr}(\Sigma AX_o^2A^T)|>t)\nonumber \\
&\leq 2\exp\Big(-c\min({L^2 t^2\over 4\|\tilde{B}\|_{\rm HS}^2 },{L t \over 2\|\tilde{B}\|_2 })\Big).
\end{align}
First note that, 
\begin{equation}
\|\tilde{B}\|_2  = \|B\|_2. 
\end{equation}

Furthermore,
\begin{align}
\|\tilde{B}\|_{\rm HS}^2&=L {\rm Tr}(B^2)=L\sum_{i=1}^m\lambda_i^2(B)\nonumber \\
&\leq L m\lambda^2_{\max}(B)=Lm\|B\|_2^2,
\end{align}
where the second equality is due to the fact that ${\rm rank}(B)= m$. On the other hand, $\| B\|_2=\|X_oA^T\Sigma AX_o\|_2\leq x^2_{\max}\|A^T\Sigma A\|_2$. Moreover,
\begin{align}
\|A^T\Sigma A\|_2^2 &=\max_{\uv\in\mathbb{R}^n}{ \uv^T A^T \Sigma A A^T\Sigma A\uv \over \|\uv\|_2^2} \nonumber \\
&\leq \lambda_{\max}(A^TA)\lambda_{\max}(AA^T)\lambda_{\max}^2(\Sigma)
\end{align}
But $\Sigma=(AX^2A^T)^{-1}$ and $X={\rm diag}(\xv)$.  Therefore, $\lambda_{\max}(\Sigma)=(\lambda_{\min}(AX^2A^T))^{-1}\leq (\lambda_{\min}(AA^T)x^2_{\min})^{-1}$ and
\begin{align*}
\|A^T\Sigma A\|_2^2\leq {\lambda_{\max}(AA^T)\lambda_{\max}(A^TA)\over \lambda^2_{\min}(AA^T) x^4_{\min}}.
\end{align*}


\subsubsection{Proof of Lemma \ref{lemma:1:upper}} \label{sec:proof:lemma:1:upper}
To prove $|\lambda_i (\tilde{\Sigma}_o^{-{1\over 2}}\Delta \Sigma \tilde{\Sigma}_o^{-{1\over 2}}) | \in [0, \frac{x_{\max}^2\lambda^2_{\max} (A A^T) \| {\xvh}_o^2 - {\xvt_o}^2\|_\infty}{x_{\min}^4 \lambda^2_{\min} (AA^T)} ]$, first note that
\begin{align}
|\lambda_i| &\leq   \frac{\sigma_{\max} (\Delta \Sigma)}{\sigma_{\min}(\tilde{\Sigma}_o)} \nonumber \\
&= \frac{\sigma_{\max} ((A\hat{X}_o^2 A^T)^{-1} -(A\tilde{X}_o^2 A^T)^{-1} )}{\sigma_{\min}(\tilde{\Sigma}_o)} \nonumber \\
&\overset{(a)}{\leq} \frac{\sigma_{\max} ((A\hat{X}_o^2 A^T)-(A\tilde{X}_o^2 A^T) )}{\sigma_{\min}(\tilde{\Sigma}_o) \sigma_{\min} (A\hat{X}_o^2 A^T) \sigma_{\min} (A\tilde{X}_o^2 A^T) }   \nonumber \\
&= \frac{\sigma_{\max}(A\tilde{X}_o^2 A^T)  \sigma_{\max} ((A\hat{X}_o^2 A^T)-(A\tilde{X}_o^2 A^T) )}{\sigma_{\min} (A\hat{X}_o^2 A^T) \sigma_{\min} (A\tilde{X}_o^2 A^T) }.  \nonumber \\
&\leq \frac{x_{\max}^2 \lambda^2_{\max}(AA^T)  \|{\xvh_o}^2 - {\xvt}_o^2\|_\infty}{x^4_{\min} \lambda_{\min}^2 (AA^T)}.
\end{align}
To obtain inequality (a) we have used Lemma \ref{lem:boundeigenvalues}. 

To prove \eqref{eq:upperconvexity} we start with
\begin{align}
{f(\hat{\Sigma}_o) -f (\tilde{\Sigma}_o) \leq |\log\det (\hat{\Sigma}_o) - \log\det (\tilde{\Sigma}_o) |} +
\frac{1}{L\sigma_w^2}|\sum_{\ell =1}^L \yv_\ell^T ((A \hat{X}_o^2A^T)^{-1} - (A \tilde{X}_o^2A^T)^{-1} ) \yv_{\ell} |.
\end{align}
We have
\begin{align}
 |&\log\det (\hat{\Sigma}_o) - \log\det (\tilde{\Sigma}_o) |\hspace{1cm} \nonumber \\
 &=|\log\det (I+\tilde{\Sigma}_o^{-{1\over 2}}\Delta \Sigma\tilde{\Sigma}_o^{-{1\over 2}}) |\nonumber\\
&\leq\sum_{i=1}^n |\log(1+\lambda_i(\tilde{\Sigma}_o^{-{1\over 2}}\Delta \Sigma \tilde{\Sigma}_o^{-{1\over 2}}))| \nonumber \\
& \overset{(b)}{\leq} 2 \sum_{i=1}^m |\lambda_i(\tilde{\Sigma}_o^{-{1\over 2}}\Delta \Sigma \tilde{\Sigma}_o^{-{1\over 2}})| \nonumber \\
&\leq 2m \frac{x_{\max}^2 \lambda^2_{\max}(AA^T)  \|{\xvh_o}^2 - {\xvt}_o^2\|_\infty}{x^4_{\min} \lambda_{\min}^2 (AA^T)},
\end{align}
where Inequality (b) uses the assumption that $\lambda_i(\tilde{\Sigma}_o^{-{1\over 2}}\Delta \Sigma \tilde{\Sigma}_o^{-{1\over 2}}) \in [-0.5, 0.5]$ and $|\log(1+x)| <= 2|x|$ when $x \in [0.5, +\infty)$.

Furthermore, note that 
\begin{align}
&\frac{1}{L\sigma_w^2}|\sum_{\ell =1}^L \yv_\ell^T ( (A \hat{X}_o^2A^T)^{-1} - (A \tilde{X}_o^2A^T)^{-1} ) \yv_{\ell} |  \nonumber \\
&\leq  \sigma_{\max } ( (A \hat{X}_o^2A^T)^{-1} - (A \tilde{X}_o^2A^T)^{-1} ) \frac{1}{L \sigma_w^2} \sum_{\ell=1}^L \yv_{\ell}^T \yv_{\ell} \nonumber \\
&\leq  \frac{ \lambda_{\max}(AA^T)  \|{\xvh_o}^2 - {\xvt}_o^2\|_\infty}{x^4_{\min} \lambda_{\min}^2 (AA^T)} \frac{1}{L \sigma_w^2} \sum_{\ell=1}^L \yv_{\ell}^T \yv_{\ell}  \nonumber \\
&\leq  \frac{ \lambda_{\max}(AA^T)  \|{\xvh_o}^2 - {\xvt}_o^2\|_\infty }{L \sigma_w^2 x^4_{\min} \lambda_{\min}^2 (AA^T)} \sum_{\ell=1}^L \yv_{\ell}^T \yv_{\ell}  \nonumber \\
&\leq \frac{x_{\max}^2 \lambda^2_{\max}(AA^T)  \|{\xvh_o}^2 - {\xvt}_o^2\|_\infty}{L \sigma_w^2 x^4_{\min} \lambda_{\min}^2 (AA^T)} \sum_{\ell=1}^L \wv_{\ell}^T \wv_{\ell}. \nonumber
\end{align}



%



\section*{Acknowledgements}
X.C., S.J., Z.H. and A.M. were supported in part by ONR award no. N00014-23-1-2371. S.J. was suppoerted in part by NSF CCF-2237538. C.A.M. was supported in part by SAAB, Inc., AFOSR Young Investigator Program Award no. FA9550-22-1-0208, and ONR award no. N00014-23-1-2752.

\bibliography{example_paper}
\bibliographystyle{icml2024}

\appendix
\onecolumn


\section{Likelihood function and its gradient} \label{app:MLE}
\subsection{Caculation of the likelihood function}
The aim of this section is to derive the loglikelihood for our model, 
\[
\yv_{\ell} = A X\wv_{\ell} + \zv_{\ell},  \ \ \ \ \ \ \ \  \ \ \ {\rm for}   \  \ \  \ell=1, \ldots, L, 
\]
where $\wv_1, \wv_2, \ldots, \wv_{L}$, and  $\zv_1, \zv_2, \ldots, \zv_L$ are independent and identically distributed $\mathcal{CN} (0, \sigma_w^2 I_n)$ and $\mathcal{CN} (0, \sigma_z^2 I_p)$ respectively. Since the noises are indepednet across the looks, we can write the loglikelihood for one of the looks, and then add the loglikelihoods to obtain the likelihood for all the looks. For notational simplicity, we write the measurements of one of the looks as: 
\[
\yv=AX\wv+\zv
\]
Note that $\yv$ is a linear combination of two Gaussian random vectors and is hence Gaussian. Hence, by writing the real and imaginary parts of $\yv$ seperately we will have
\[
\Re(\yv) + \Im(\yv) = (\Re(AX) + i \Im(AX))(\wv^{(1)} + i\wv^{(2)}) + (\zv^{(1)} + i\zv^{(2)}),
\]
and
\begin{align*}
\tilde{\yv} \triangleq \begin{bmatrix} \Re (\yv) \\ \Im (\yv) \end{bmatrix} &\sim \mathcal{N} \left( \begin{bmatrix} 0 \\ 0 \end{bmatrix},  B  \right),
\end{align*}
where
\begin{align*}
    B = \begin{bmatrix} \sigma_z^2 I_n + {\sigma_w^2}\Re (A X^2 \bar{A}^T) & -\sigma_w^2 \Im (A X^2 \bar{A}^T) \\ \sigma_w^2 \Im (A X^2 \bar{A}^T) & \sigma_z^2 I_n + \sigma_w^2 \Re (A X^2 \bar{A}^T) \end{bmatrix}.
\end{align*}

Hence, the log-likelihood of our data $\yv$ as a function of $\xv$ is
\begin{align}\label{eq:ll-SL1}
    l(\xv) =& -\frac{1}{2} \log \det \left( B \right) - \frac{1}{2} 
    \begin{bmatrix}
        \Re (\yv^T) & \Im (\yv^T)
    \end{bmatrix} \left(  B \right)^{-1} \begin{bmatrix} \Re (\yv) \\ \Im (\yv) \end{bmatrix} + C. 
\end{align}

Note that equation~\eqref{eq:ll-SL1} is for a single look. Hence the loglikelihood of $\yv_1, \yv_2, \ldots, \yv_L$ as a function of $\xv$ is:
\begin{align}
    l(\xv) = -\frac{L}{2}  \log \det(B) -\frac{1}{2} \sum_{l=1}^L
    \tilde{\yv}_l^T B^{-1} \tilde{\yv}_l + C, 
\end{align}
Since we would like to maximize $l(\xv)$ as a function of $\xv$, for notational simplicty we define the cost function $f_L(\xv): \mathbb R^n \to \mathbb R$:
\begin{align}
    f_L(\xv) = \log \det (B) + \frac{1}{L } \sum_{l=1}^L \tilde{\yv}_l^T B^{-1} \tilde{\yv}_l,
\end{align}
that we will minimize to obtain the maximum likelihood estimate. 

\subsection{Calculation of the gradient of the likelihood function}

As discussed in the main text, to execute the projected gradient descent, it is necessary to compute the gradient of the negative log-likelihood function $\partial f_L$.
The derivatives of $f_L$ with respect to each element $\xv_j$ of $\xv$ is given by:
\begin{align}
    \frac{\partial f_L}{\partial \xv_j} =& {2 \xv_j \sigma_w^2} \left( \begin{bmatrix} \Re (\av_{\cdot, j}^T) & \Im (\av_{\cdot, j}^T) \end{bmatrix}
    B^{-1} 
    \begin{bmatrix} \Re (\av_{\cdot, j}) \\ \Im (\av_{\cdot, j}) \end{bmatrix} \right. + \left.
    \begin{bmatrix} -\Im (\av_{\cdot, j}^T) & \Re (\av_{\cdot, j}^T) \end{bmatrix}
    B^{-1} 
    \begin{bmatrix} -\Im (\av_{\cdot, j}) \\ \Re (\av_{\cdot, j}) \end{bmatrix} \right) \nonumber \\
    &- \frac{2 \xv_j \sigma_w^2}{L } \sum_{l=1}^L \left[ \left(
    \begin{bmatrix} \Re (\av_{\cdot,j}^T) & \Im (\av_{\cdot,j}^T) \end{bmatrix} B^{-1}
    \begin{bmatrix} \Re (\yv_l) \\ \Im (\yv_l) \end{bmatrix} \right)^2 \right. + \left.
    \left(
    \begin{bmatrix} -\Im (\av_{\cdot,j}^T) & \Re (\av_{\cdot, j}^T) \end{bmatrix}
    B^{-1} 
    \begin{bmatrix} \Re (\yv_l) \\ \Im (\yv_l) \end{bmatrix}
    \right)^2 \right] \nonumber \\
    =& 2 \xv_j \sigma_w^2 \left( \tilde{\av}_{\cdot, j}^{+T} B^{-1} \tilde{\av}_{\cdot, j}^{+} + \tilde{\av}_{\cdot, j}^{-T} B^{-1} \tilde{\av}_{\cdot, j}^{-} \right) - \frac{2 \xv_j \sigma_w^2}{L } \sum_{l=1}^L \left[ \left( \tilde{\av}_{\cdot, j}^{+T} B^{-1} \tilde{\yv}_l \right)^2 + \left( \tilde{\av}_{\cdot, j}^{-T} B^{-1} \tilde{\yv}_l \right)^2 \right], \label{eq:derivative}
\end{align}
where $\av_{\cdot,j}$ denotes the $j$-th column of matrix $A$, $\tilde{\av}_{\cdot, j}^+ = \begin{bmatrix}
    \Re (\av_{\cdot,j}) \\ \Im (\av_{\cdot, j})
\end{bmatrix}$ and $\tilde{\av}_{\cdot, j}^- = \begin{bmatrix}
    -\Im (\av_{\cdot, j}) \\ \Re (\av_{\cdot, j})
\end{bmatrix}$.

\subsection{More simplification of the gradient} \label{ssec:matrixinversioncalc}
The special form of the matrix $B$ enables us to do the calculations more efficiently. To see this point, define:
\begin{align*}
    U + iV \triangleq \left( \sigma_z^2 I_n + {\sigma_w^2} A X^2 \bar{A}^T \right)^{-1},
\end{align*}
where $U, V \in R^{m \times m}$. These two matrices should satisfy:
\begin{align*}
    \left( \sigma_z^2 I_n + {\sigma_w^2} \Re (A X^2 \bar{A}^T) \right) U -{\sigma_w^2} \Im (A X^2 \bar{A}^T) V = I_n \\
    {\sigma_w^2} \Im (A X^2 \bar{A}^T) U + \left( \sigma_z^2 I_n + {\sigma_w^2}\Re (A X^2 \bar{A}^T) \right) V = 0.
\end{align*}
These two equations imply that:
\begin{align}\label{eq:Bsymmetric}
    B^{-1} = \begin{bmatrix} U & -V \\ V & U \end{bmatrix}.
\end{align}

This simple observation, enables us to reduce the number of multiplications required for the Newton-Schulz algorithm. More specifically, instead of requiring to multiply two $2m \times 2m$ matrices, we can do $4$ multiplications of $m \times m$ matrices. This helps us have a factor of $2$ reduction in the cost of matrix-matrix multiplication in our Newton-Schulz algorithm. 

In cases the exact inverse calculation is required, again this property enables us to reduce the inversion of matrix $B \in \mathbb{R}^{2m \times 2m}$ to the inversion of two $m \times m$ matrices (albeit a few $m\times m$ matrix multiplications are required as well). 

Plugging \eqref{eq:Bsymmetric} into~\eqref{eq:derivative}, we obtain a simplified form for the gradient of $f_L(\xv)$:
\begin{align}
    \frac{\partial f_L}{\partial \xv_j} =& {4 \xv_j \sigma_w^2}  \Re\left(\bar{\av}_{\cdot, j}^T (U + iV) \av_{\cdot, j}\right) - \frac{2 \xv_j \sigma_w^2}{L} \sum_{l=1}^L \big[ \Re^2 \left( \bar{\av}_{\cdot, j}^T (U + iV) \yv_l \right) + \Im^2 \left( \bar{\av}_{\cdot, j}^T (U + iV) \yv_l \right) \big] \nonumber \\
    =& {4 \xv_j \sigma_w^2} \Re\left(\bar{\av}_{\cdot, j}^T (U + iV) \av_{\cdot, j} \right) - \frac{2 \xv_j \sigma_w^2}{L } \sum_{l=1}^L \left\lVert \bar{\av}_{\cdot, j}^T (U + iV) \yv_l \right\rVert_2^2.
\end{align}

\section{Details of our Bagged-DIP-based PGD}
\label{app:add_experiment}
 Algorithm~\ref{alg:PGD} shows a detailed version of the final algorithm we execute for recovering images from their multilook, speckle-corrupted, undersampled measurements. In one of the steps of the algorithm we ensure that all the pixel values of our estimate are within the range $[0,1]$. This is because we have assumed that the image pixels take values within $[0,1]$. 
 
\begin{algorithm}[ht]
    \small
  \caption{Iterative PGD algorithm}
  \label{alg:PGD}
  \begin{algorithmic}
    \STATE {\bfseries Input:} $\{\mathbf{y}_l\}^L_{l=1}, A, \mathbf{x}_0 = \frac{1}{L} \sum^L_{l=1} |A^T \mathbf{y}_l|, g_{\theta}(\cdot)$.
    \STATE {\bfseries Output:} Reconstructed $\hat{\mathbf{x}}$.
    \FOR{$t=1,~\ldots, ~T$}
        \STATE \textbf{[Gradient Descent Step]}
        \IF {$t=1$ \textbf{or} $\|\xv^{t}-\xv^{t-1}\|_{\infty} > \delta_{\xv}$}
            \STATE Calculate exact $B_{t} = (AX^2_{t}A^T)^{-1}$.
        \ELSE
            \STATE Approx $\Tilde{B}_{t} = B_{t-1} + B_{t-1}(I_m - AX^2_{t}A^T B_{t-1})$.
        \ENDIF
        \STATE Gradient calculation at coordinate $j$ as $\nabla f(\xv_{t-1,j})$ using $B_{t}$ or $\Tilde{B}_{t}$, and update $\xv^G_{t,j}$: $\xv^{G}_{t,j} \leftarrow \xv_{t-1,j} - \mu_t \nabla f(\xv_{t-1,j})$. 
        \STATE Save matrix inverse $B_t$ or $\Tilde{B}_{t}$.
        \STATE Truncate $\xv^G_{t}$ into range $(0,1)$, $\xv^G_{t} = \text{clip}(\xv^G_{t}, 0, 1)$.
        \STATE \textbf{[Bagged-DIPs Projection Step]}
        \STATE Generate random image given randomly generated noise $\uv \sim \mathcal{N}(0,1)$ as $g_{\theta}(\uv)$.
        \STATE Update $\theta_t$ by optimizing over $\| g_{\theta}(\uv) - \xv^G_t \|^2_2$: $ \theta_t \leftarrow \operatorname*{argmin}_{\theta} \| g_{\theta}(\uv) - \mathbf{x}^G_t \|^2_2$ till converges.
        \STATE Generate $\mathbf{x}^P_t$ using trained $g_{\hat{\theta}_{t}}(\cdot)$ as $\xv^P_t \leftarrow g_{\theta_{t}}(\uv)$.
        \STATE Obtain $\xv_t = \xv^{P}_t $.
    \ENDFOR
    \STATE Reconstruct image as $\hat{\mathbf{x}} = \mathbf{x}_T$.
  \end{algorithmic}
\end{algorithm}

The only remaining parts of the algorithm we need to clarify are, (1) our hyperparameter choices, and (2) the implementation details of the Bagged-DIP module. As described in the main text, in each (outer) iteration of PGD, we learn three DIPs and then take the average of their outputs. Let us now consider one of these DIPs that is applied to one of the $h_k \times w_k$ patches. 

Inspired by the deep decoder paper \cite{heckel2018deep}, we construct our neural network, using four blocks: we call the first three blocks DIP-blocks and the last one output block. 
The strctures of the blocks are shown in Figure~\ref{fig:DIP_block_structure}. As is clear from the figrue, each DIP block is composed of the following components:
\begin{itemize}
\item Up sample: This unit increases the hight and width of the datacube that receives by a factor of 2. To interpolate the missing elements, it uses the simple bilinear interpolation. Hence, if the size of the image is $128 \times 128$, then the height and width of the input to DIP-block3 will be $64 \times 64$, the input of DIP-Block2 will be $32 \times 32$, and so on.
\item ReLU: this module is quite standard and does not require further explanation. 

\item Convolution: For all our simulations we have either used $1\times 1$ or $3 \times 3$ convolutions. Additionally, we provide details on the number of channels for the data cubes entering each block in our simulations. The channel numbers are [128, 128, 128, 128] for the four blocks.

\end{itemize}

The output block is simpler than the other three blocks. It only have a 2D convolution that uses the same size as the convolutions of the other DIP blocks. The nonlinearity used here is sigmoid, since we assume that the pixel values are between $[0,1]$.

Finally, we should mention that each element of the input noise $\uv$ of DIP (as described before DIP function is $g_{\thetav} (\uv)$) is generated independently from Normal distribution $\mathcal{N}(0,1)$. 

\begin{figure*}[ht]
\centering
\includegraphics[width=0.99\textwidth]{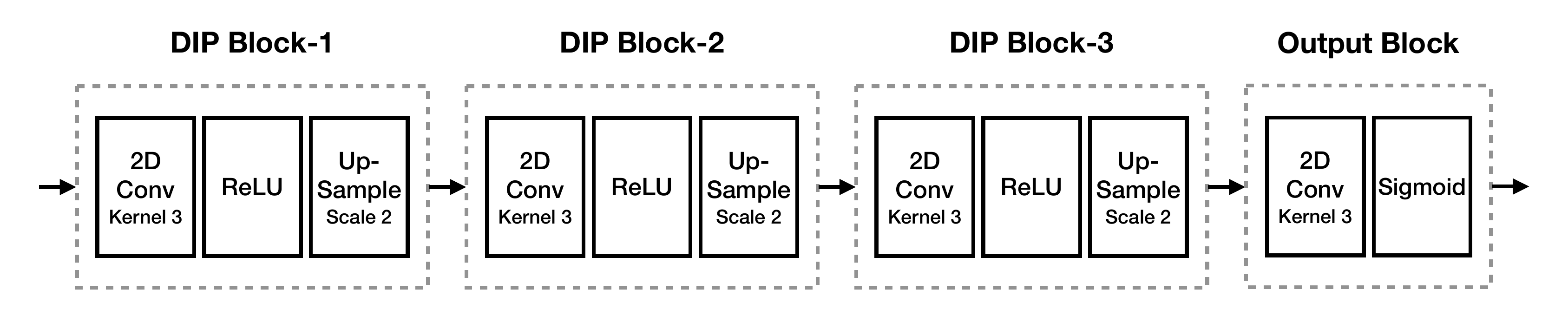}
\caption{The structure of DIP and Output Blocks.}
\label{fig:DIP_block_structure}
\end{figure*}

The other hyperparameters that are used in the DIP-based PGD algorithm are set in the following way: The learning rate of the loglikelihood gradient descent step (in PGD) is set to $\mu=0.01$. For training the Bagged-DIPs, we use Adam~\cite{kingma2014adam} with the learning rate set to $0.001$ and weight decay  set to $0$. The number of iterations used for training Bagged-DIPs for different estimates on images are mentioned in Table~\ref{tab:num_iteration_DIP}. We run the outer loop (gradient descent of likelihood) for $100,200,300$ iterations when $m/n=0.5,0.25,0.125$ respectively. For ``Cameraman`` only, when $m/n=0.125$, since the convergence rate is slow, we run $800$ outer iterations.

\begin{table*}[ht]
    \centering
    \scriptsize
\begin{tabular}{lccccccccc}\hline
     &\textbf{Patch size of estimates} & \textbf{Barbara} & \textbf{Peppers} & \textbf{House} & \textbf{Foreman} & \textbf{Boats} & \textbf{Parrots} & \textbf{Cameraman} & \textbf{Monarch}\\\hline

    & 128 &400  &400  &400  &400  &400  &800  &4k  &800  \\
    & 64 &300  &300  &300  &300  &300  &600  &2k  &600  \\
    & 32 &200  &200  &200  &200  &200  &400  &1k  &400  \\\hline
\end{tabular}
    \caption{Number of iterations used in training Bagged-DIPs for different estimates.}
    \label{tab:num_iteration_DIP}
\vspace{-0.1in}
\end{table*}

The Newton-Schulz algorithm, utilized for approximating the inverse of matrix \( B_t \), has a quadratic convergence when the maximum singular value \( \sigma_{\max}(I - M^0 B_{t}) < 1 \). Hence, ideally, if this condition does not hold, we do not want to use the Newton-Schulz algorithm, and may prefer the exact inversion. Unfortunately, checking the condition \( \sigma_{\max}(I - M^0 B_{t}) < 1 \) is also computationally demanding. However, the special form of $B_t$ enables us to check have easier heuristic evaluation of this condition. 

For our problems, we establish an empirical sufficient condition for convergence: \( \|\xv_t - \xv_{t-1}\|_{\infty} < \delta_{\xv} \), where \( \delta_{\xv} \) is a predetermined constant. To determine the most robust value for \( \delta_{\xv} \), we conducted simple experiments. We set \( n = 128 \times 128 \) and \( m / n = 0.5 \). The sensing matrix \( A \) is generated as described in the main part of the paper (see Section \ref{sec:simulation_baggedDIP}). Each element of \( {\xv}_o \) is independently drawn from a uniform distribution \( U[0.001, 1] \). Furthermore, each element of \( \Delta {\xv}_o \) is independently sampled from a two-point distribution. In this distribution, the probability of the variable \( X \) being \( \delta_{\xv} \) is equal to the probability of \( X \) being \( -\delta_{\xv} \), both with a probability of 0.5, ensuring \( \|\Delta \xv_o\|_{\infty} = \delta_{\xv} \). We define \( B \) as \( A (X + \Delta X_o)^2 \bar{A}^T \), and \( M^0 \) as \( (A X^2 \bar{A}^T)^{-1} \). We then assess the convergence of the Newton-Schulz algorithm for calculating \( B^{-1} \). For various values of \( \delta_{\xv} \), we ran the simulation 100 times each, recording the convergence success rate. As indicated in Table \ref{tab:threshold}, the algorithm demonstrates instability when \( \delta_{\xv} \geq 0.13 \). Consequently, we set \( \delta_{\xv} \) to 0.12 in all our simulations to ensure the reliable convergence of the Newton-Schulz algorithm.

\begin{figure}[t]
    \centering
    \includegraphics[width=0.5\textwidth]{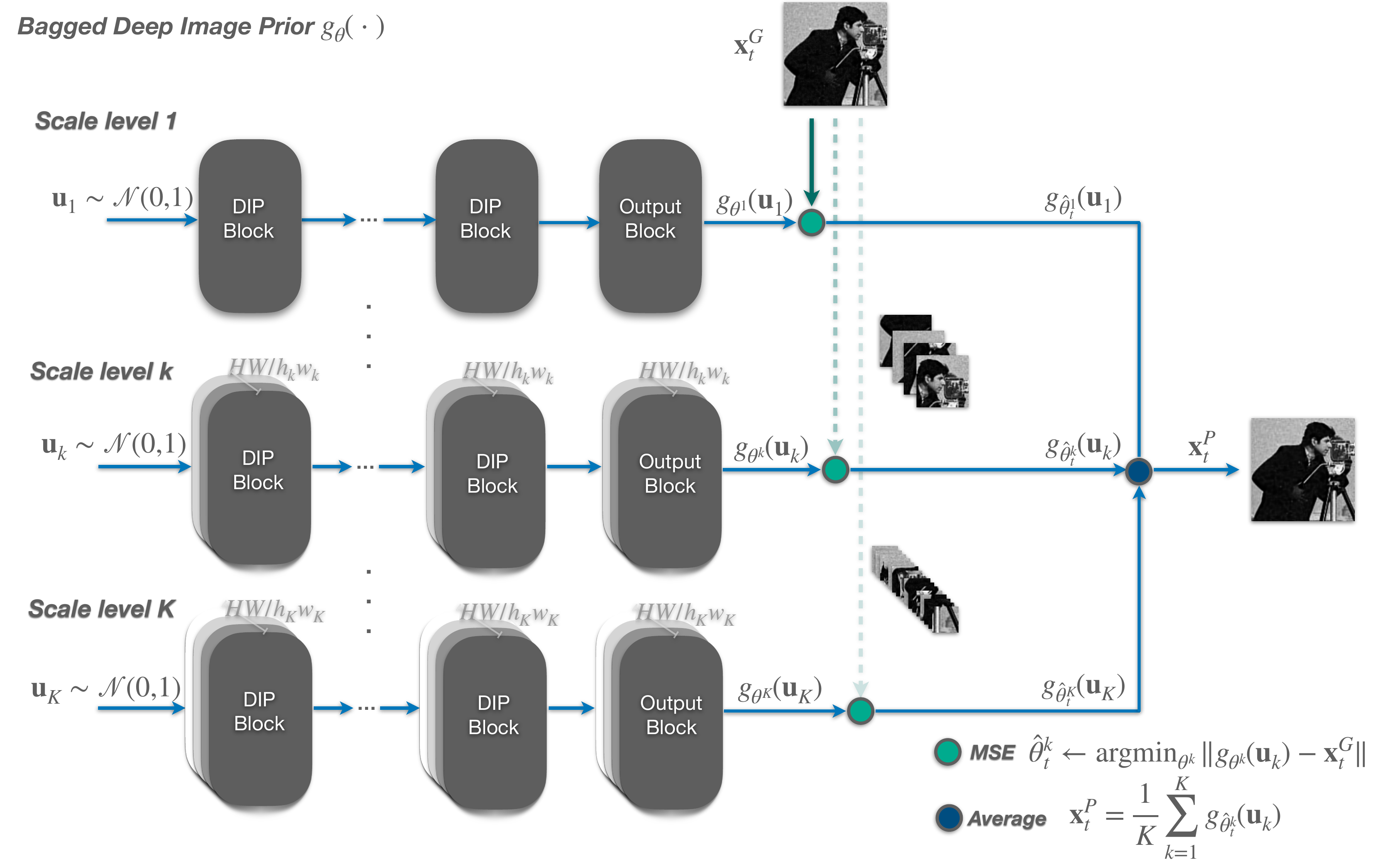}
    \includegraphics[width=0.43\textwidth]{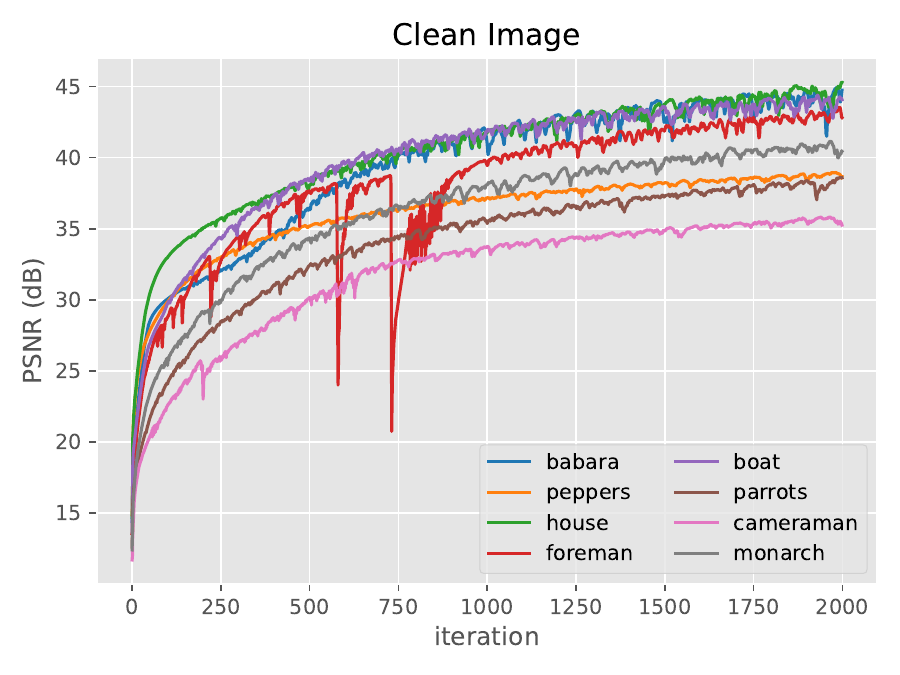}
    \caption{(Left) The structure of Bagged-DIPs with $K$ estimates. (Right) Performance of fitting Bagged-DIP to clean images.}
    \label{fig:model_structure}
\end{figure}

\begin{table}[ht]
    \centering
    \scriptsize
\begin{tabular}{ccc}\hline
    & $\delta_{\xv}$ & \textbf{convergence success rate} \\\hline
    & 0.1 & 100\% \\
    & 0.11 & 100\% \\
    & 0.12 & 100\% \\
    & 0.13 & 38\% \\
    & 0.14 & 0\% \\
    & 0.15 & 0\% \\\hline
\end{tabular}
\caption{Convergence success rate for different thresholds.}
\label{tab:threshold}
\end{table}

\section{Additional experiments.}
\subsection{Comparison with classical despeckling.}
\label{app:add_exp_DnCNN}

We use DnCNN~\cite{zhang2017learning,zhang2017beyond} as the neural networks for despeckling task. The DnCNN structure we use consists of one input block, eight DnCNN block and one output block. The details are shown in Figure~\ref{fig:DnCNN_block_structure}. The number of channels for each convolutional layer is 64.

\begin{figure*}[ht]
\centering
\includegraphics[width=0.99\textwidth]{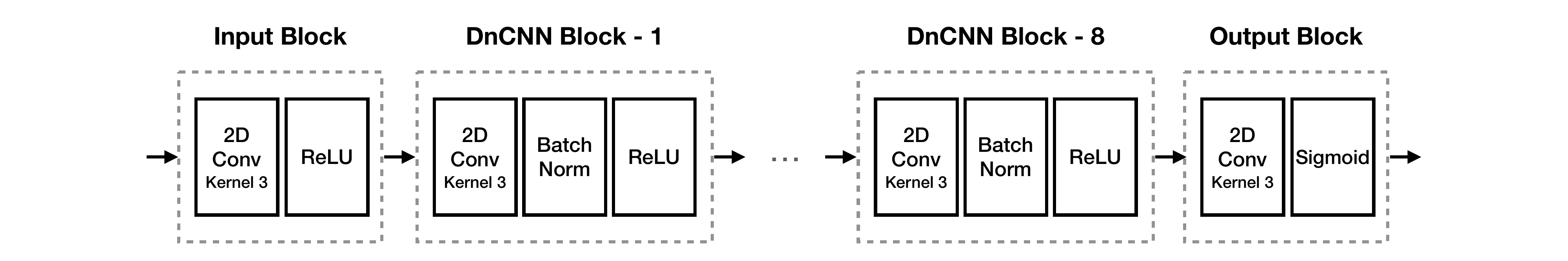}
\caption{The structure of DnCNN, Input and Output Blocks.}
\label{fig:DnCNN_block_structure}
\end{figure*}
Since there is no sensing matrix $A$ in this simulation, we consider real-valued speckle noise $\wv$ (similar to what we considered in our theoretical work) to make the number of measurements the same (with the same number of looks), and make the comparisons simpler. The training set we use for DnCNN is BSD400~\cite{martin2001database}, we divide the images in training set to be $128 \times 128$, with stride step 32. The learning rate is 1e-4, batch size is 64, it takes 20 epochs for the training to converge. Table~\ref{tab:DnCNN_UB} compares the results of this simulation, with the results of the simulation we presented in the main text for fifty percent downsampled measurement matrix. The results of DnCNN-UB are often between 1-3dB better than the results of our Bagged-DIPs-based PGD. With the exception of the cameraman, where our Bagged-DIPs-based PGD seems to be better. It should be noted that the gain obtained by DnCNN here should not be only associated to the fact that the matrix $A$ is undersampled. There is one major differences that may be contributing to the improvements that we see in Table~\ref{tab:DnCNN_UB}, and that is we are also using a training set, while our DIP-based method does not use any training data.

\begin{table*}[ht]
    \centering
    \scriptsize
\begin{tabular}{lccccccccccc}\hline
    \textbf{m/n} & \textbf{$\#$looks} & \textbf{Barbara} & \textbf{Peppers} & \textbf{House} & \textbf{Foreman} & \textbf{Boats} & \textbf{Parrots} & \textbf{Cameraman} & \textbf{Monarch} & \textbf{Average}\\\hline
\multirow{3}{*}{50\%}  
    & 25 & 27.30/0.759 & 27.02/0.724 & 28.56/0.697 & 27.56/0.735 & 26.21/0.669 & 25.94/0.728 & 27.95/0.762 & 27.17/0.845 & 27.21/0.740\\
    & 50 & 28.67/0.816 & 28.52/0.804 & 30.30/0.762 & 28.88/0.827 & 27.58/0.739 & 27.23/0.799 & 30.21/0.843 & 28.86/0.898 & 28.78/0.818\\
    & 100 & 29.40/0.843 & 29.21/0.849 & 31.61/0.815 & 29.74/0.871 & 28.45/0.785 & 28.20/0.848 & 31.58/0.902 & 30.05/0.932 & 29.78/0.856\\\hline\hline
\multirow{3}{*}{DnCNN-UB}
    & 25 & 28.77/0.832 & 29.42/0.830 & 29.90/0.799 & 31.26/0.820 & 28.51/0.788 & 28.19/0.858 & 28.77/0.911 & 29.40/0.925 & 29.28/0.845\\
    & 50 & 30.30/0.873 & 30.93/0.868 & 30.92/0.832 & 31.54/0.859 & 29.87/0.839 & 29.07/0.888 & 29.84/0.925 & 30.63/0.940 & 30.40/0.878\\
    & 100 & 32.18/0.903 & 32.38/0.899 & 32.45/0.866 & 34.28/0.900 & 31.59/0.865 & 29.77/0.902 & 29.83/0.896 & 31.27/0.955 & 31.72/0.898\\\hline
\end{tabular}
    \caption{PSNR(dB)/SSIM $\uparrow$ of $m/n=50\%$, $L=25/50/100$. DnCNN despeckling tasks are used for showing the performance gap between our method with $50\%$ downsampled complex valued measurements and the corresponding empirical uppper bound.}
    \label{tab:DnCNN_UB}
\vspace{-0.1in}
\end{table*}

\subsection{Comparison with \cite{chen2023multilook}}
\label{app:add_exp_baseline_DIP}

As discussed before, a recent paper has also considered the problem of recovering an image from undersampled/ill-conditioned measurement kernels in the presence of the speckle noise \cite{chen2023multilook} and they also considered a DIP-based approach. The goal of this section is to provide some comparison in the performance of our DIP-based method and the one presented in \cite{chen2023multilook}. 

While \cite{chen2023multilook} considered real-valued speckle noises and measurements, we adapt their approch to the complex-valued settings under which we have run our experiments. The results of our comparisons are presented in Table~\ref{tab:baselines_comparison}. As is clear in this table \cite{chen2023multilook}  considered two different algorithms DIP-simple and DIP-$M^3$. DIP-simple is the DIP-based PGD with filter size = 1 and the number of channels were chosen as $[100,50,25,10]$. In DIP-$M^3$, the same DIP was chosen. But the authors also used $\lambda$ residual connection to balance the contribution from gradient descent and projection outputs as follows:
\begin{align*}
    \xv_t = \lambda \xv^P_t + (1-\lambda)\xv^G_t,
\end{align*}
where $\xv^P_t$ and $\xv^G_t$ are gradient descent and projection results respectively. Similar to the setting of that paper we consider, where smaller $\lambda$ is used when $L$ increase, so we set the hyperparameter $\lambda=0.3,0.2,0.1$ for $L=25,50,100$. 

We should note that choosing optimal $\lambda$ is tricky for different $m/n$ and $L$. Setting a small $\lambda$ means that a large portion of projection has been bypassed, which indicates the limit on learning capacity of simple DIP. To verify the statement that a simple DIP has very limited learning capacity on some images, we also provide the baseline, DIP-simple, which distingushes from DIP-$M^3$ by setting $\lambda=1.0$. This means we use the projection results fully from the DIP-simple, and it fails on several tasks as we can see from Table~\ref{tab:baselines_comparison}.

It's important to highlight that Bagged-DIPs do not incorporate residual connections to bypass the projection, essentially representing the case where $\lambda=1.0$. Specifically, we solely rely on the projection from Bagged-DIPs. The outcomes presented in Table~\ref{tab:baselines_comparison} demonstrate the robust projection capabilities of Bagged-DIPs, leading to superior performance compared to DIP-simple and DIP-$M^3$.

\begin{table}[ht]
    \centering
    \scriptsize
\begin{tabular}{lccccc}\hline
    \textbf{m/n} & \textbf{$\#$looks} & \textbf{DIP-simple} & \textbf{DIP-$M^3$} & \textbf{Bagged-DIPs} \\\hline
\multirow{3}{*}{12.5\%}
    & 25 & 18.03/0.336 & 17.81/0.316 & \textbf{19.24/0.406}\\
    & 50 & 19.25/0.408 & 18.96/0.382 & \textbf{20.83/0.538}\\
    & 100 & 20.15/0.497 & 19.94/0.464 & \textbf{21.78/0.612}\\\hline\hline
\multirow{3}{*}{25\%}
    & 25 & 22.00/0.493 & 21.69/0.474 & \textbf{22.86/0.549}\\
    & 50 & 23.42/0.572 & 23.39/0.551 & \textbf{24.95/0.672}\\
    & 100 & 24.79/0.656 & 25.08/0.629 & \textbf{26.24/0.745}\\\hline\hline
\multirow{3}{*}{50\%}
    & 25 & 25.62/0.683 & 26.01/0.668 & \textbf{27.21/0.740}\\
    & 50 & 26.81/0.749 & 27.81/0.733 & \textbf{28.78/0.818}\\
    & 100 & 27.53/0.799 & 29.52/0.779 & \textbf{29.78/0.856}\\\hline
\end{tabular}
    \caption{Average PSNR(dB)/SSIM $\uparrow$ comparison of baseline methods, $m/n=12.5\%/25\%/50\%$, $L=25/50/100$.}
    \label{tab:baselines_comparison}
\vspace{-0.1in}
\end{table}

\subsection{Bagging performance.}
\label{app:add_exp_bagging}
We show the comparison of Bagged-DIPs with three sophisticated DIP estimates in Figure~\ref{fig:baselines_S50_L50}. In most of the test images, the bagging of three estimates yields a performance improvement over all individual estimates. As described in the main part of the paper, this is expected when all the estimates are low-bias and are weakly-dependent. 

However, there are exceptions, such as the case of ``Foreman." In such instances, one of the estimates appears to surpass our average estimate. This occurs when certain individual estimates are affected by large biases. While these biases are mitigated to some extent in our average estimate, a residual portion persists, affecting the overall performance of the average estimate. There are a few directions one can explore to resolve this issue and we leave them for future research. For instance, we can create more bagging samples, and then use more complicated networks without worrying about the overfitting. That will alleviate the issue of high bias that exists in a few images.


\begin{figure}[t]
\centering
\includegraphics[width=0.23\textwidth]{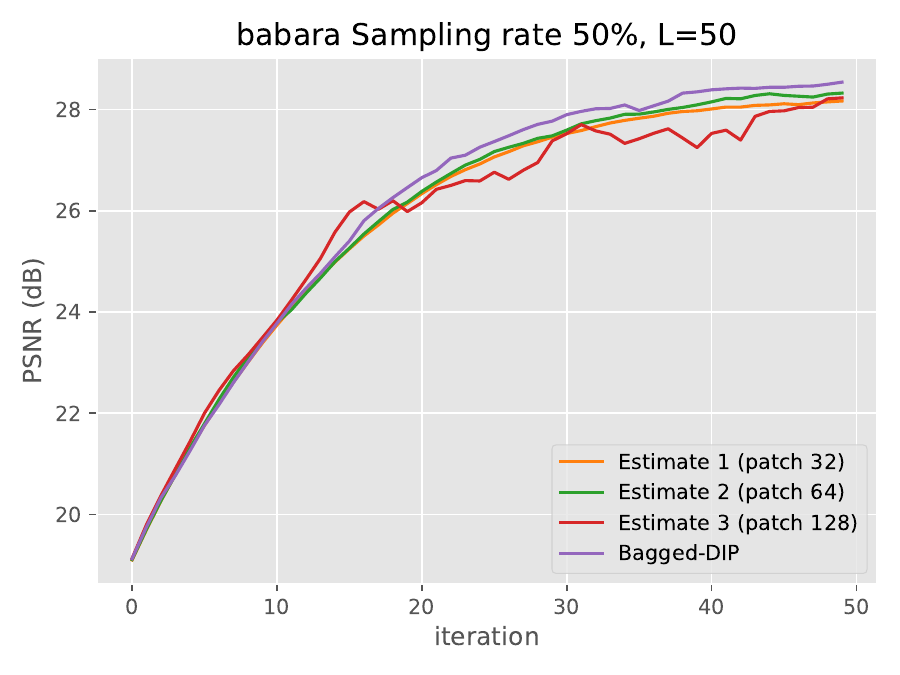}
\includegraphics[width=0.23\textwidth]{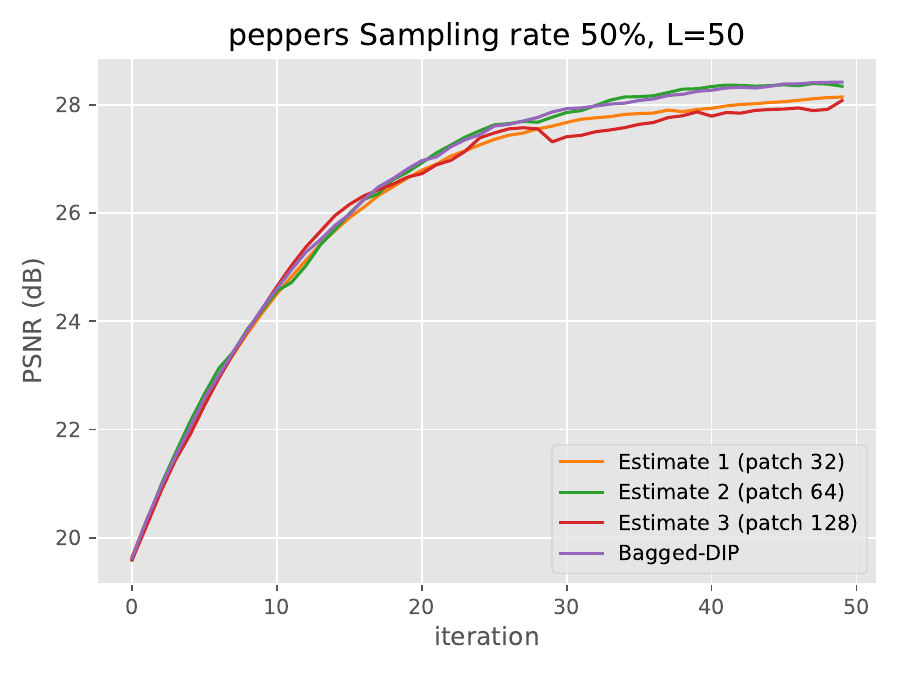}
\includegraphics[width=0.23\textwidth]{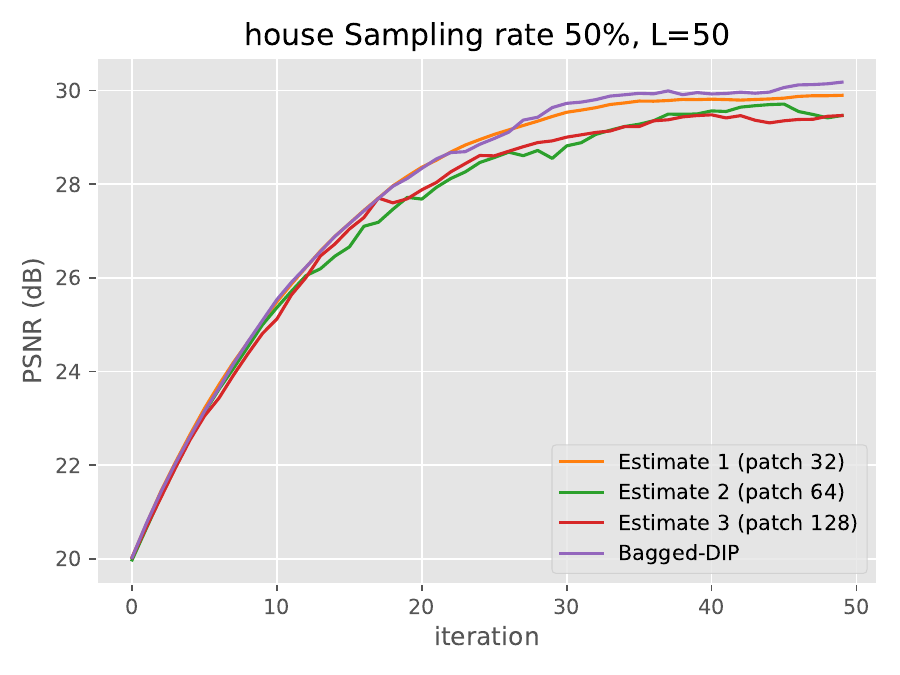}
\includegraphics[width=0.23\textwidth]{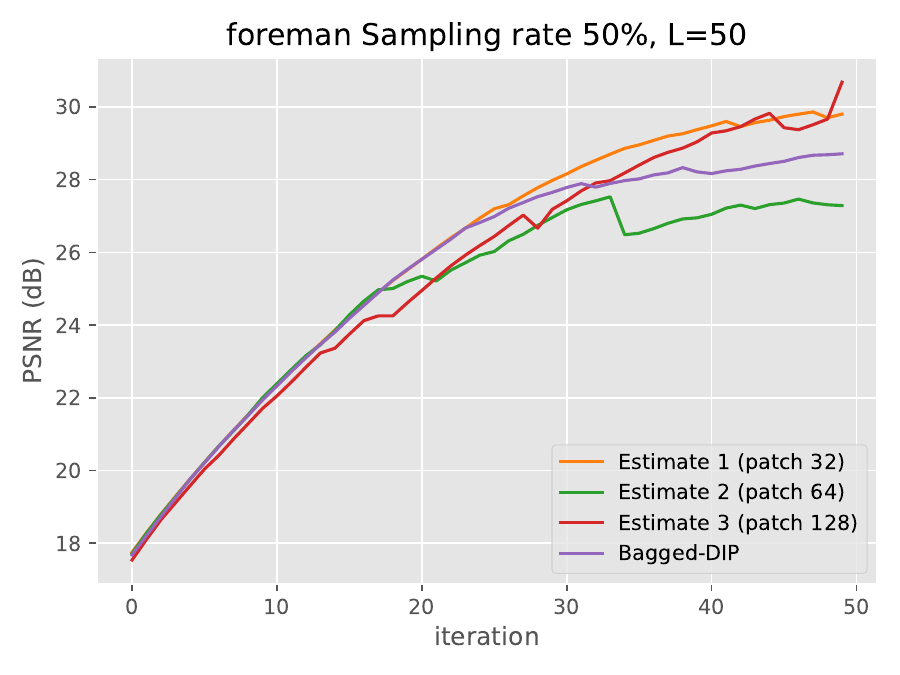}
\includegraphics[width=0.23\textwidth]{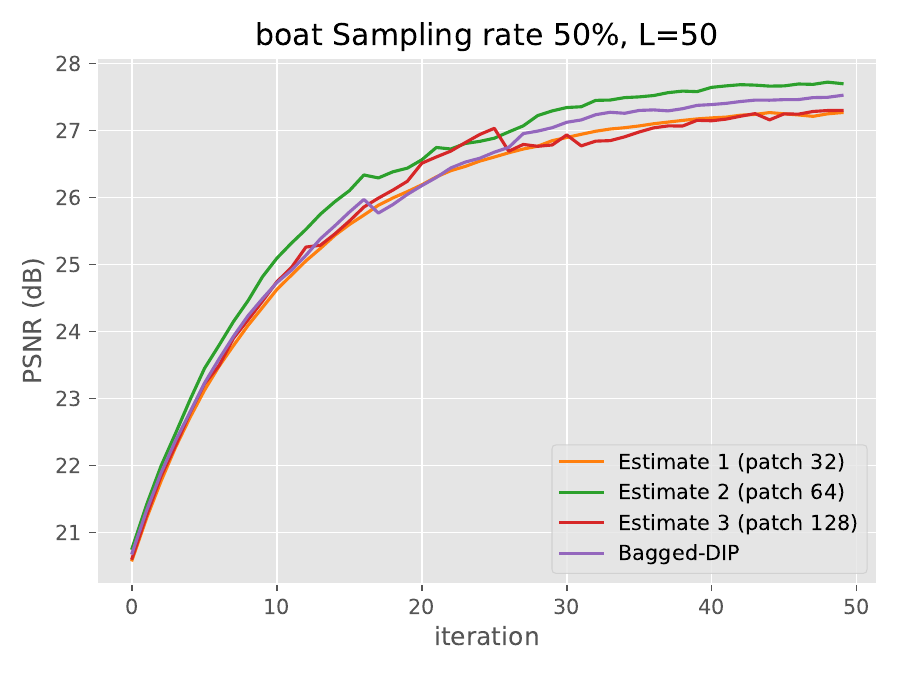}
\includegraphics[width=0.23\textwidth]{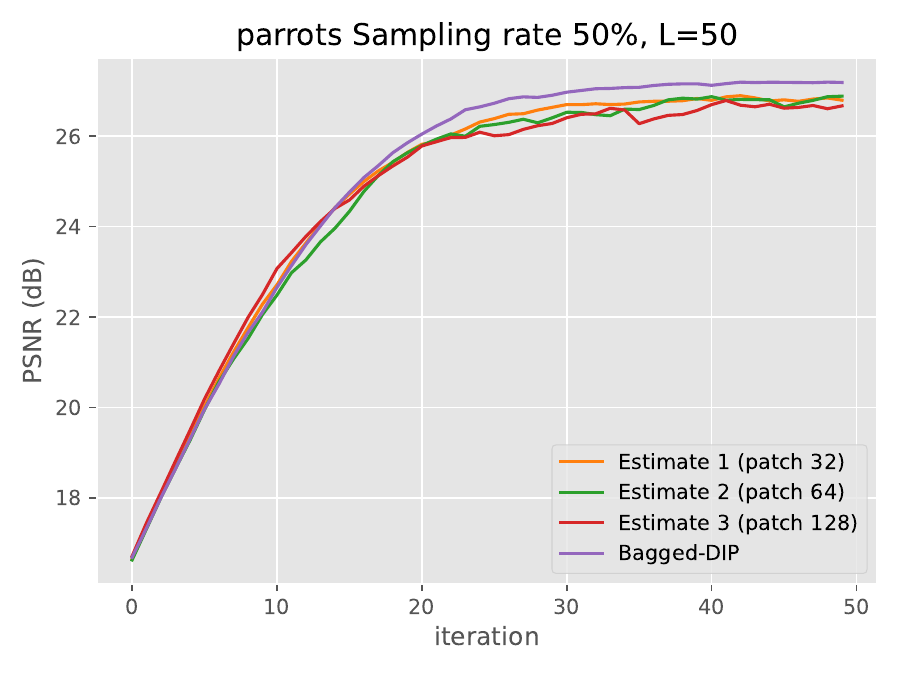}
\includegraphics[width=0.23\textwidth]{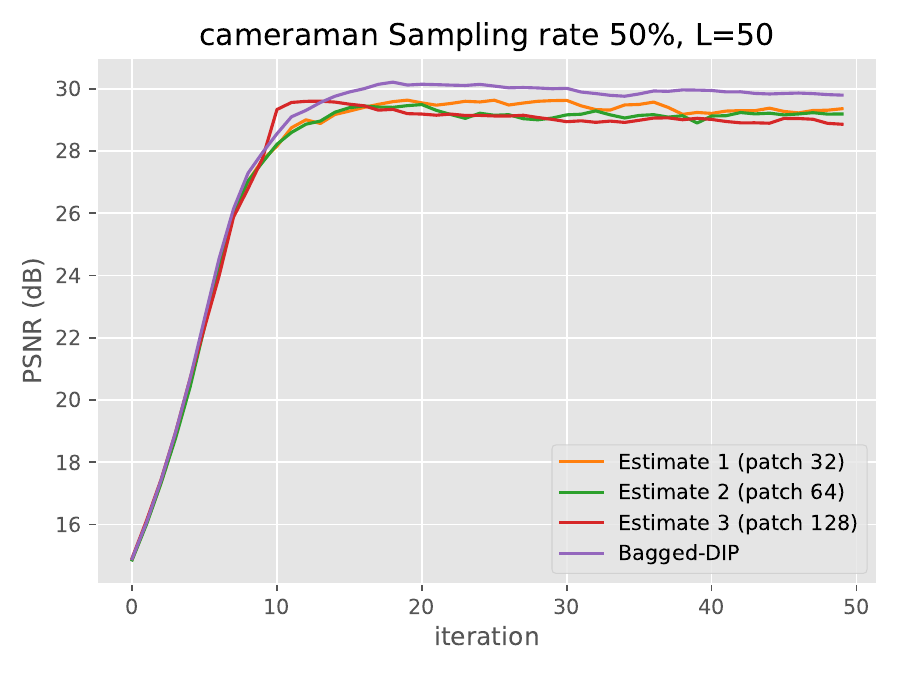}
\includegraphics[width=0.23\textwidth]{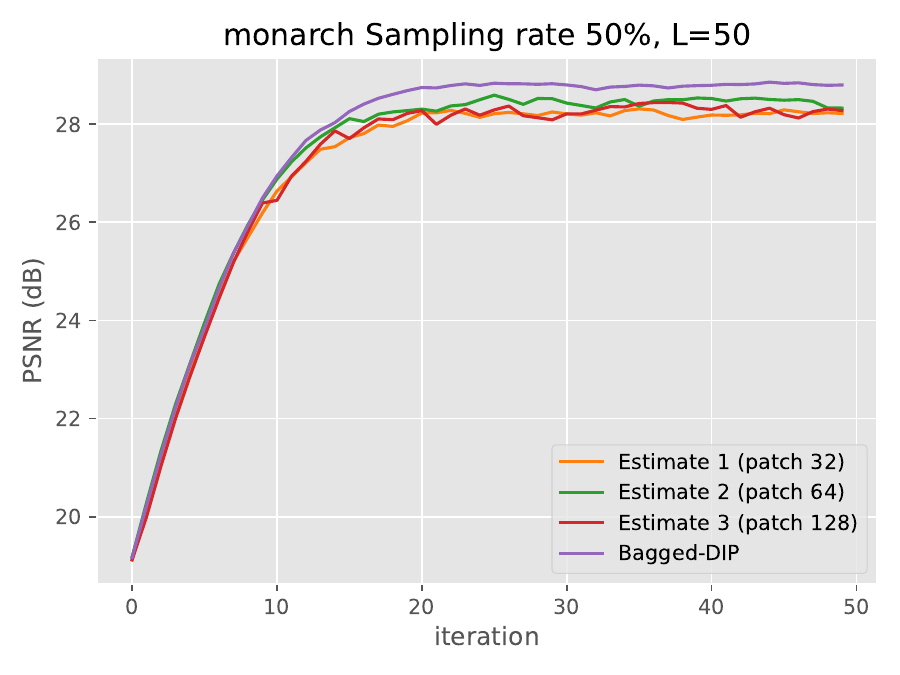}
\caption{Comparison of Bagged-DIPs and three sophisticated DIP estimates on 8 images, $m/n=0.5$, $L=50$.}
\label{fig:baselines_S50_L50}
\end{figure}

\section{Time cost of PGD algorithm.}\label{sec:initialization:comp}
We provide the timing details of training the Bagged-DIPs-based PGD algorithm in Table~\ref{tab:PGD_time}. The time cost of each iteration in PGD is affected by $m/n$, and the iterations needed for training Bagged-DIPs. The experiments are performed on Nvidia RTX 6000 GPUs, and we record the time it uses accordingly for different tasks.

\begin{table}[ht]
    \centering
    \scriptsize
\begin{tabular}{ccccc}\hline
    &\textbf{Bagged-DIPs training iterations} & \textbf{12.5\%} & \textbf{25\%} & \textbf{50\%} \\\hline
    &200, 300, 400 & $\sim$ 65  & $\sim$ 75  & $\sim$ 105 \\
    &400, 600, 800 & $\sim$ 115  & $\sim$ 125  & $\sim$ 155 \\
    &1k, 2k, 4k & $\sim$ 330  & $\sim$ 340  & $\sim$ 370 \\\hline
\end{tabular}
\caption{Time (in seconds) required for each iteration of the PGD with different sampling rate and iterations for training Bagged-DIPs.}
\label{tab:PGD_time}
\end{table}

We also find that, compared with intialize $\mathbf{x}_0$ with fixed values, using initialization $\mathbf{x}_0 = \frac{1}{L} \sum^L_{l=1} |\Bar{A}^T \mathbf{y}_l|$ helps improve the convergence rate. But the final reconstructed performance does not depend on the initialization methods we compare. The effect of input intialization in PGD algorithm is shown in Figure~\ref{fig:init_compare}.

\begin{figure}[ht]
\centering
\includegraphics[width=0.3\textwidth]{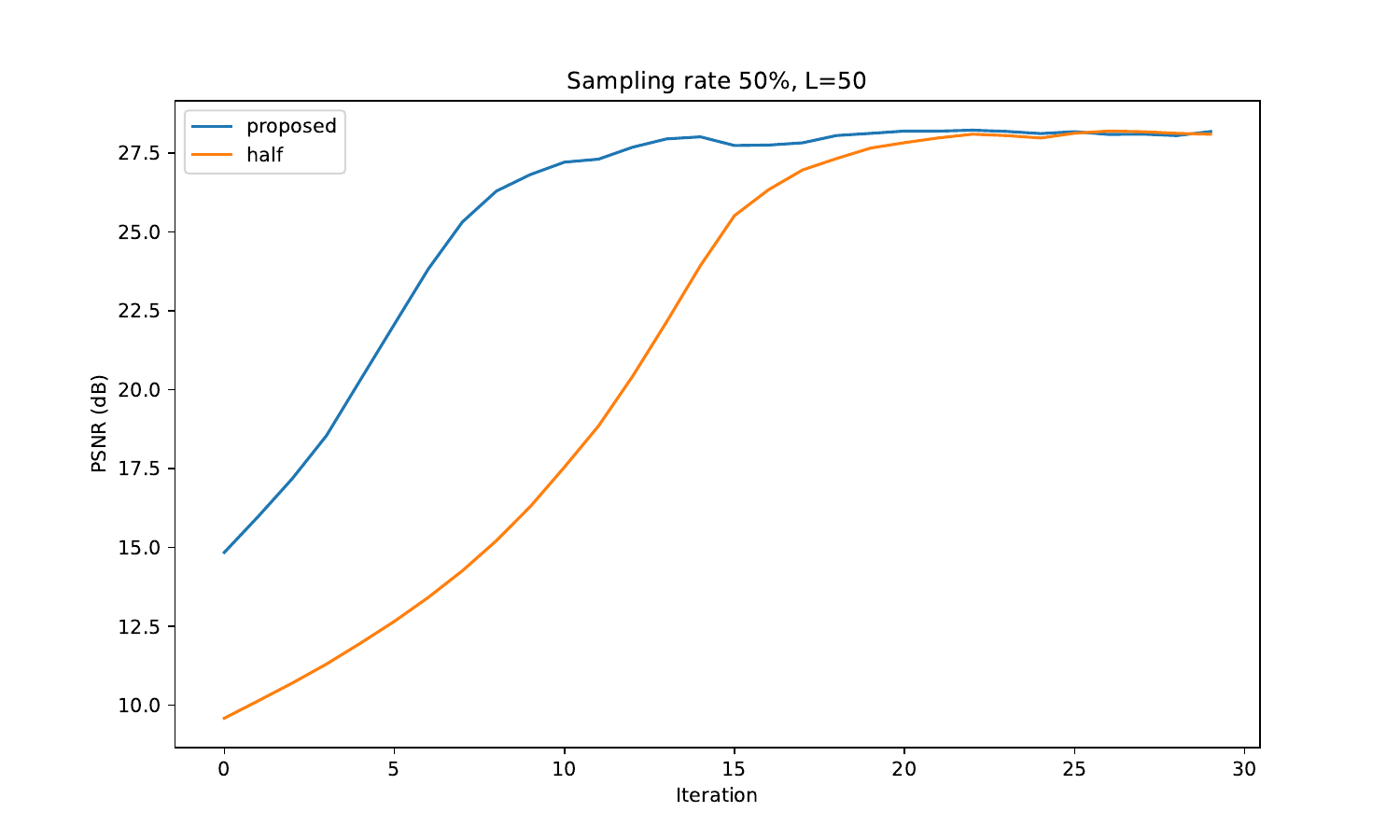}
\vspace{-10pt}
\caption{Comparison between two initializations: (1) proposed: our inititalization method; (2) half: a vector whose elements are all $0.5$.}
\label{fig:init_compare}
\end{figure}

\section{Qualitative results of Bagged-DIPs-based PGD}\label{sec:qualitative_results}
We show the qualitative results of the Bagged-DIPs-based PGD algorithm in Figure~\ref{fig:qualitative_vis}. The images shown are reconstructed from $L=25, 50, 100$ looks of $12.5, 25, 50 \%$ downsampled complex-valued measurements. Row 1-3 are $m/n=0.125$ with $L=25,50,100$ respectively, row 4-6 are $m/n=0.25$ with $L=25,50,100$ respectively, row 7-9 are $m/n=0.5$ with $L=25,50,100$ respectively. We also report the PSNR/SSIM under each reconstructed image. 

An evident advantage observed in the reconstructed images is that utilizing bagging estimates with various patch sizes helps alleviate the blocking issue that may arise when relying solely on a single patch size (e.g., patch size = 32). This is because boundaries in smaller patch sizes do not necessarily align with boundaries in larger patch sizes, and aggregating estimates from different patch sizes through averaging can effectively mitigate the blocking issue.

\begin{figure}[ht]
\centering
\includegraphics[width=0.45\textwidth]{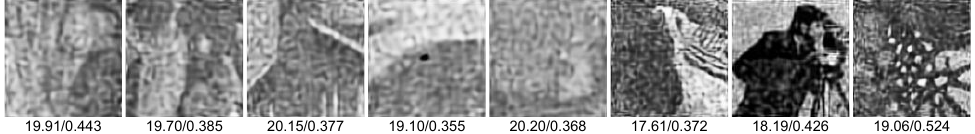}\\
\includegraphics[width=0.45\textwidth]{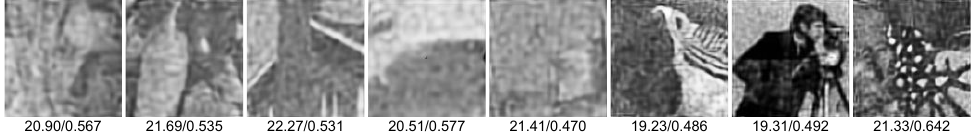}\\
\includegraphics[width=0.45\textwidth]{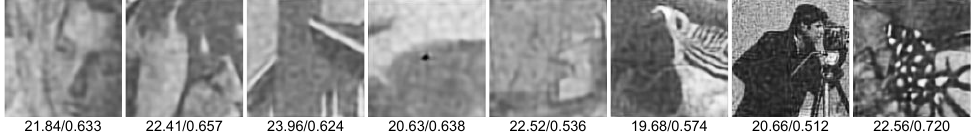}\\
\includegraphics[width=0.45\textwidth]{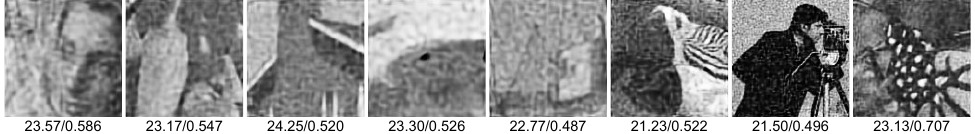}\\
\includegraphics[width=0.45\textwidth]{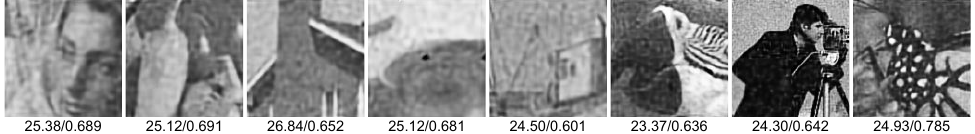}\\
\includegraphics[width=0.45\textwidth]{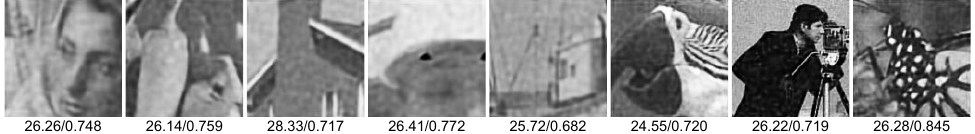}\\
\includegraphics[width=0.45\textwidth]{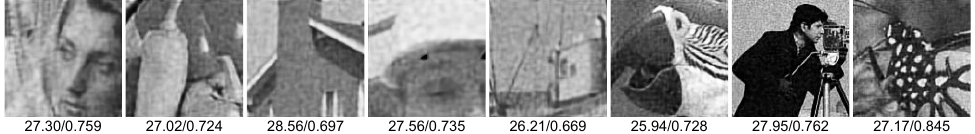}\\
\includegraphics[width=0.45\textwidth]{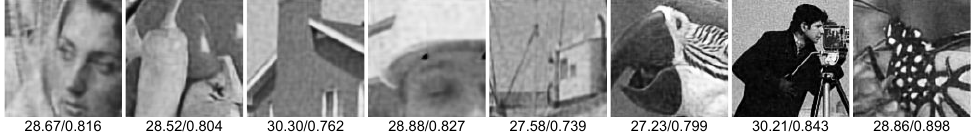}\\
\includegraphics[width=0.45\textwidth]{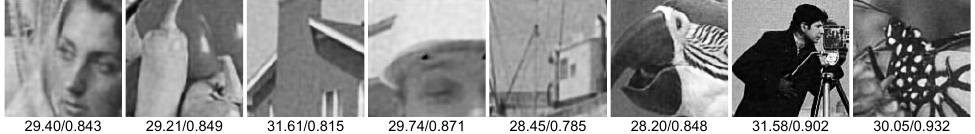}
\caption{Reconstructed images from $L=25, 50, 100$ looks of $12.5, 25, 50 \%$ downsampled complex-valued measurements. Row 1-3 are $m/n=0.125$ with $L=25,50,100$ respectively, row 4-6 are $m/n=0.25$ with $L=25,50,100$ respectively, row 7-9 are $m/n=0.5$ with $L=25,50,100$ respectively}
\label{fig:qualitative_vis}
\end{figure}


\end{document}